\documentclass[letterpaper, 11pt]{article}
\usepackage{latexsym}
\usepackage{amssymb}
\usepackage[margin=1in]{geometry}
\usepackage{amsmath,amsfonts,amsthm}
\usepackage{graphicx}
\usepackage{multirow}
\usepackage{ifsym}
\usepackage{color}
\usepackage{hyperref}
\usepackage{dsfont}
\newcommand{\ble}{\begin{lemma}}
\newcommand{\ele}{\end{lemma}}

\theoremstyle{plain}
\newtheorem{lemma}{Lemma}[section]

\newtheorem{theorem}[lemma]{Theorem}

\newtheorem{corollary}[lemma]{Corollary}

\newtheorem{claim}[lemma]{Claim}

\newtheorem{algorithm}[lemma]{Algorithm}

\theoremstyle{definition}
\newtheorem{definition}[lemma]{Definition}

\newcommand{\beao}{\begin{eqnarray*}}
\newcommand{\eeao}{\end{eqnarray*}\noindent}
\newcommand{\beam}{\begin{eqnarray}}
\newcommand{\eeam}{\end{eqnarray}\noindent}

\newcounter{note}[section]

\begin{document}

\begin{titlepage}

\title{Approximating Approximate Distance Oracles}

\author{Michael Dinitz\\Department of Computer Science\\Johns Hopkins University \and Zeyu Zhang\\Department of Computer Science\\Johns Hopkins University}

\date{}

\maketitle

\thispagestyle{empty}

\begin{abstract}
Given a finite metric space $(V,d)$, an approximate distance oracle is a data structure which, when queried on two points $u,v \in V$, returns an approximation to the the actual distance between $u$ and $v$ which is within some bounded stretch factor of the true distance.  There has been significant work on the tradeoff between the important parameters of approximate distance oracles (and in particular between the size, stretch, and query time), but in this paper we take a different point of view, that of per-instance optimization.  If we are given an particular input metric space and stretch bound, can we find the smallest possible approximate distance oracle for that particular input?  Since this question is not even well-defined, we restrict our attention to well-known classes of approximate distance oracles, and study whether we can optimize over those classes. 

In particular, we give an $O(\log n)$-approximation to the problem of finding the smallest stretch $3$ Thorup-Zwick distance oracle, as well as the problem of finding the smallest P\v{a}tra\c{s}cu-Roditty distance oracle. We also prove a matching $\Omega(\log n)$ lower bound for both problems, and an $\Omega(n^{\frac{1}{k}-\frac{1}{2^{k-1}}})$ integrality gap for the more general stretch $(2k-1)$ Thorup-Zwick distance oracle. We also consider the problem of approximating the best TZ or PR approximate distance oracle \emph{with outliers}, and show that more advanced techniques (SDP relaxations in particular) allow us to optimize even in the presence of outliers.  
\end{abstract}

\end{titlepage}

\section{Introduction} \label{sec:intro}

Given a finite metric space $(V,d)$, an approximate distance oracle is a data structure which can approximately answer distance queries. It is usually a combination of a preprocessing algorithm to compute a data structure, and a query algorithm which returns a distance $d'(u,v)$ whenever queried on a pair of vertices $u,v\in V$. An approximate distance oracle is said the have \emph{stretch $t$} if $d(u,v)\le d'(u,v)\le t\cdot d(u,v)$.  Note that there is a trivial stretch $1$ distance oracle that uses $\Theta(n^2)$ space: we could just store the entire metric space.  So the goal is to reduce the space, i.e., to build a small data structure that also has small stretch and small query time.    

The seminal work on approximate distance oracles is due to Thorup and Zwick~\cite{TZ05}. They showed that for every integer $k \geq 1$, every finite metric space has an approximate distance oracle with stretch $(2k-1)$ and query time $O(k)$ which uses only $O(kn^{1+\frac{1}{k}})$ space.  A significant fraction of more recent results have built off of the ideas developed in~\cite{TZ05}, and much of this follow-up work has stored the exact same (or very similar) data structure, just with improved query algorithms or slightly different information in the storage (see, e.g.,~\cite{PRT12,Wulff13,Chechik14,Chechik15}).  Most notably, P\v{a}tra\c{s}cu and Roditty~\cite{PR10} gave a different distance oracle (still using some of the basic ideas from~\cite{TZ05}) that has multiplicative stretch of $2$ and additive stretch of $1$, with size $O(n^\frac{5}{3})$.  This broke through the stretch $3$ barrier from~\cite{TZ05}.  Later this result was improved to more general multiplicative/additive stretches~\cite{AG11}.

%Thorup-Zwick's work was later generalized to fractional stretches by~\cite{PRT12}. Meanwhile~\cite{MN06},~\cite{Wulff13},~\cite{Chechik14} also did some improvement on query time, and~\cite{Chechik15} slightly improved the size of the data structure by removing multiplier $k$. Besides these results, P\v{a}tra\c{s}cu-Roditty~\cite{PR10} gave a different distance oracle that has $2$ multiplicative stretch and $1$ additive stretch, with size $O(n^\frac{5}{3})$. Later this result was improved to general stretches~\cite{AG11}. \mdnote{Move this paragraph to related work}

In this paper we ask a natural but very different type of question about approximate distance oracles: can we find (or approximate) the \emph{best} approximate distance oracle?  If we are given an input metric space and a stretch bound, is it possible to find the \emph{smallest} approximate distance oracle for that particular input?  This is an unusual question in two ways.  First, most data structures are by design forced to store all of the input data; the question is \emph{how} to store it and what \emph{extra} information should be stored.  This is the case in other settings where instance-optimality of data structures has been considered, e.g., static or dynamic optimality of splay trees.  Second, it is not clear whether this question is even well-defined: lower bounds on data structures are commonly arrived at through information or communication complexity (see, e.g., \cite{Jacobson88}) but when we ask for the optimal data structure on one particular instance this approach becomes meaningless.

However, approximate distance oracles are different in ways which allow us to make meaningful progress towards these optimization questions.  First, since we are allowed to return only approximate distances (up to some stretch factor), we are allowed to store only part of the input (and indeed this is the entire point of such an oracle).  The second problem is a bit more tricky: given an input, how can we optimize over ``the space of all approximate distance oracles"?  What does this mean, and what does this space look like?  

To get around this issue, we make an observation: many modern distance oracles (and in particular Thorup-Zwick, P\v{a}tra\c{s}cu-Roditty, and almost all of their variants) have a similar structure. The preprocessing algorithm chooses a subset of the original distances to store which has some particular structure, and the query algorithm can return a valid distance estimate efficiently as long as the stored distances satisfy the required structure. Thus we can optimize for \emph{these particular} distance oracles by choosing the \emph{best possible} set of distances to remember subject to the required structure.  By characterizing this structure for different types of distance oracles, we can optimize over those types.  

% because we are allowed to lose some information from the original data. Thus we are able to choose which to remember and which to forget. This solves the first problem. Moreover, we can fix a query algorithm which already has good performance and is widely used, so that we can define an optimal solution. This solves the second problem.

%This problem is trickier than the spanner case because of the solution space. Unlike spanners which are just subgraphs of the original graph, distance oracles are made up of preprocessing and query algorithms. It is not at all clear how to optimize over the space of all relevant algorithms, which seems significantly more complex than the space of all subgraphs.

%We noticed that many modern distance oracles (and in particular Thorup-Zwick, P\v{a}tra\c{s}cu-Roditty, and almost all of its variants) have a similar structure: the preprocessing algorithm chooses a subset of the original distances to remember which has some particular structure, and the query algorithm can return a valid distance estimate efficiently as long as the remembered distances satisfy the required structure. Thus we can optimize for \emph{these particular} distance oracles by choosing the \emph{best possible} set of distances to remember subject to the required structure.

For example, the stretch-$3$ Thorup-Zwick distance oracle uses a subtle but simple method to choose the set of distances to store.  It randomly samples a subset of approximately $\sqrt{n}$ vertices, without using any information about the original metric space, and then creates a data structure which is related (in a well-defined, important way) to these vertices. The correctness of the query algorithm does not depend on the choice of the vertices.  Thus instead of simply choosing the subset of vertices uniformly at random, we can instead try to optimize the set of chosen vertices with respect to the actual input metric space. 

% as centers to record the distance to all other vertices. This turns out to be very lossy when the original metric space is easy to record. In this paper, we will talk about how to optimize a distance oracle while the required structure and the query algorithm are fixed.

In this paper, we give matching $\Theta(\log n)$ upper and lower bounds for optimizing stretch-$3$ Thorup-Zwick distance oracles, and matching $\Theta(\log n)$ upper and lower bounds for optimizing the P\v{a}tra\c{s}cu-Roditty distance oracle. These upper bounds both use a similar LP relaxation, but by giving an $\Omega(n^{\frac{1}{k}-\frac{1}{2^{k-1}}})$ integrality gap for optimizing stretch-$(2k-1)$ Thorup-Zwick distance oracles, we show that this relaxation is not enough to give nontrivial approximations when extended to larger stretch values.

As an extension, we also study the problem of optimizing distance oracles \emph{with outliers}: if we are allowed to not answer queries for some of the vertices (of our choosing), can we have much smaller storage space? We give an $(O(\log n),1+\varepsilon)$-bicriteria approximation to both stretch-$3$ Thorup-Zwick and P\v{a}tra\c{s}cu-Roditty distance oracles with outliers. We also give a true approximation to stretch-$3$ Thorup-Zwick distance oracle with outliers when the number of outliers is small.

\paragraph{Relationship to Spanners.} 
It is worth noting that this paper is motivated by a similar line of research on \emph{graph spanners} (subgraphs which approximately preserve distances).  Spanners and distance oracles tend to be related (although there is no known formal connection between them), and the traditional questions asked of spanners (what is the tradeoff between the stretch and the size?) are similar to the traditional questions asked of distance oracles.  Recently, there has been significant progress in looking at spanners from an optimization point of view: given an input graph and an allowed stretch bound, can we find the sparsest possible spanner meeting that stretch bound?  In the last few years, upper and lower bounds have been developed for these problems in the basic case, the directed case, with a degree objective, with fault-tolerance, etc.  See, e.g., \cite{DK11,BBMRY13,DZ16,CDK12,CD14}.

It is natural to ask these kinds of optimization questions for distance oracles as well, but the definitions become much more difficult.  For spanners, the space we are optimizing over (all subgraphs) is very clear and well-defined.  But for distance oracles, as discussed, it is much harder to define the space of all data structures.  Thus in this paper we optimize over restricted classes, where this space is more well-defined.  We view our definitions of these restricted optimization questions as one of the major contributions of this work.  

\section{Definitions and Preliminaries}

We begin with some basic definitions, including formal definitions of the problems that we will be working on.

\begin{definition}
An approximate distance oracle with $(m,a)$-stretch, size $s$, preprocessing time $g$, and query time $h$ is a pair of algorithms, $preprocess$ and $query$, with the following properties.
\begin{itemize}
\item $preprocess$ is a randomized preprocessing algorithm $preprocess(V,d,m,a,r)$ which takes as input a metric space $(V,d)$, stretch bound $(m,a)$, and random string $r$ and outputs a data structure $S$ where the expected output size is at most $\mathbb{E}_r[|S|]\le s(|V|,m,a)$ and the 
expected preprocessing time is at most $g(|V|,m,a)$.
\item $query$ takes as input a data structure $S=preprocess(V,d,m,a,r)$ (the output of the preprocess algorithm) with two vertices $u,v \in V$, and outputs a value $d'(u,v) \in \mathbb{R}$ such that $d(u,v)\le d'(u,v)\le m\cdot d(u,v)+a$. The running time of $query$ is at most $h(|V|,m,a)$.
\end{itemize}
\end{definition}

We will frequently refer to these just as ``distance oracles" rather than ``approximate distance oracles" when the stretch bound is clear from context.

The query algorithm guarantees here are deterministic: the randomness only affects the size of the data structure. Note that one could easily define distance oracles so that either the correctness (with respect to the stretch bound) or the query running time (or both) hold only in expectation or with high probability, but as discussed in Section~\ref{sec:intro}, essentially all existing distance oracles (and in particular the Thorup-Zwick distance oracle) have deterministic guarantees on the queries.  

This naturally leads us to the following question: If we fix a particular distance oracle and metric space, can we find the \emph{best} possible data structure?  Here we will focus on the output size, not the preprocessing time (as long as the preprocessing time is polynomial).  In other words, since the query algorithm work on \emph{any} of the possible data structures which the preprocessing algorithm might output, can we actually find the smallest such data structure?  This gives the following natural optimization problem.

\begin{definition}
Given an approximate distance oracle $\mathcal{A}=(preprocess,query)$, the $\mathcal{A}$-optimization problem takes as input a metric space $(V,d)$ and a stretch bound $(m,a)$, and the goal is to find a string $r$ which minimizes $|preprocess(V,d,m,a,r)|$.
\end{definition}

In this paper we will focus on two distance oracles (Thorup-Zwick~\cite{TZ05} and  P\v{a}tra\c{s}cu-Roditty~\cite{PR10}), so we now introduce these oracles.

\subsection{Thorup-Zwick Distance Oracle}

For every integer $k \geq 1$, Thorup and Zwick~\cite{TZ05} provided an approximate distance oracle with $(2k-1,0)$-stretch, size $O(n^{1+\frac{1}{k}})$, preprocessing time $O(kn^{2+\frac{1}{k}})$, and query time $O(k)$. We call this distance oracle $TZ_k$.

Their preprocessing algorithm first constructs a chain of subsets $\varnothing=A_k\subseteq A_{k-1}\subseteq\mathellipsis\subseteq A_0=V$ by repeated sampling. Each set $A_i$, where $i\in[k-1]$, is obtained by including each element of $A_{i-1}$ independently with probability $n^{-\frac{1}{k}}$.

Let $R_{iu}=\{v\in A_{i-1}\mid d(u,v)<\min_{w\in A_i}d(u,w)\}$ for all $u\in V$ and $i\in[k]$ (where by convention $\min_{w\in\varnothing}d(u,w)=\infty$ for all $u\in V$ to handle the $i=k$ case). The output data structure is obtained by storing (in a $2$-level hash table) the distance from each node $u$ to each node in $\bigcup_{i=1}^kR_{iu}$.

% a specific set of vertices determined by the A$_i$ subsets. In particular, the data structure stores the distance from $u$ to the vertices in $\bigcup_{i=1}^k\{v\in A_{i-1}\backslash\{u\}\mid d(u,v)<\min_{w\in A_i}d(u,w)\}$, which is, for every level $i$, $u$ remembers the distance between $u$ and the vertices in $A_{i-1}$ which is closer to $u$ than any vertex in $A_i$.

The data structure also stores a little more information. Each vertex $u$ remembers $k-1$ pivots: $\arg\min_{w\in A_i}d(u,w)$ for all $i\in[k-1]$, and the distance from $u$ to these pivots. However, this is a fixed space cost, and also negligible, so when analyzing the size of the oracle we will ignore the cost of storing the pivots

Clearly the output data structure is determined once $A_1,\mathellipsis,A_{k-1}$ are fixed. The size of the data structure is:
\[cost(A_1,\mathellipsis,A_{k-1},V,d)=\sum_{u\in V}\sum_{i=1}^k|R_{iu}|=\sum_{u\in V}\sum_{i=1}^k\left|\{v\in A_{i-1}\mid d(u,v)<\min_{w\in A_i}d(u,w)\}\right|.\]
We will refer to  $\sum_{u\in V}|R_{iu}|$ as the cost in level $i$.

Let us look back on the definition of approximate distance oracle. The random string $r$ is only used to generate $A_i$'s, and the query algorithm will return a correct distance estimate no matter what the sets $A_i$ are, but the size is determined by the sets.  Therefore, the $TZ_k$-optimization problem is to find the subsets $\varnothing=A_k\subseteq A_{k-1}\subseteq\mathellipsis\subseteq A_0=V$ in order to minimize the total cost.

\subsection{P\v{a}tra\c{s}cu-Roditty Distance Oracle}

P\v{a}tra\c{s}cu and Roditty~\cite{PR10} provided an approximate distance oracle with $(2,1)$-stretch, size $O(n^\frac{5}{3})$, preprocessing time $O(n^2)$, and query time $O(1)$. We call this distance oracle $PR$. Note that $PR$ works only for metric spaces with integer distances.

Their preprocessing algorithm first construct a set $A\subseteq V$ via a complicated correlated sampling (informally, they sample a large set and a small set, and then define $A$ to be everything in the large set and everything contained in a ball around the small set delimited by the large set). The data structure consists of a $2$-level hash table for the distance from each node in $A$ to each node in $V$, as well as a $2$-level hash table storing the distance between each pair $\{u,v\}\subseteq V$ such that $d(u,v)<\min_{w\in A}d(u,w)+\min_{w\in A}d(v,w)-1$

%they remember the distance between each vertex in $A$ and every other vertices in $V$. In addition, if a pair of vertex $(u,v)$ is close enough, say , they also remember the distance between $(u,v)$. \mdnote{This should be made more formal.  Tone is a bit too conversational for a formal definition -- probably want to change it to be more along the lines of ``Let $A$ be a set sampled as follows...  The data structure consists of a 2-level hash table for the distance from each node in $A$ to each node in $V$, as well as a 2-level hash table storing the distance between each pair $u,v \in V$ such that $d(u,v) < $...}

As with Thorup-Zwick, the output data structure is completely determined once $A$ is fixed. Let $R=\left\{\{u,v\}\subseteq V\mid d(u,v)<\min_{w\in A}d(u,w)+\min_{w\in A}d(v,w)-1\right\}$.  Then the size of the data structure is
\[cost(A,V,d)=n\cdot|A|+|R|=n\cdot|A|+\left|\left\{\{u,v\}\subseteq V\mid d(u,v)<\min_{w\in A}d(u,w)+\min_{w\in A}d(v,w)-1\right\}\right|.\]
As before, the random string $r$ is only used to generate the set $A$, and any $A\subseteq V$ gives a data structure on which the query algorithm works. Therefore, the $PR$-optimization problem is to find the subset $A\subseteq V$ in order to minimize the total cost.

\subsection{Distance Oracles With Outliers}

In some cases, a small set of outlier vertices may make the size of the data structure blow up.  Yet in some applications it is acceptable to ignore these outliers.  This was the motivation behind a line of work on distance oracles with slack (\cite{CDGKS09},~\cite{CDG06}), in which the data structure could ignore the stretch bound on a small fraction of the distances.

In this paper, we consider the case that we can refuse to answer distance queries for some outlier vertices.  In other words, we can essentially remove an outlier set $F$ out of $V$ when computing the distance oracle.  This gives us the problem of optimizing distance oracle with outliers, in which we not only need to find the random string to determine the output data structure, we also need to find the set of outliers to minimize the final cost.  More formally, we have the following type of problem.

%We can easily extend the concept of optimizing distance oracles to the case that we allow some outliers. It is possible that we do not need to know the approximate distance of \textsl{all} pair of vertices. We just need to answer the distance queries except an outlier set $F$. \mdnote{Need to more formally define an approximate distance oracle with outliers, and discuss (presumably after the definition) how it contrasts with related notions such as distance oracles with slack.  Or maybe don't need formal definition (since it's just a normal DO with some points removed), but still need to state that more clearly.}

\begin{definition}
Given an approximate distance oracle $\mathcal{A}=(preprocess,query)$, the $\mathcal{A}$-optimization problem with outliers takes as input a metric space $(V,d)$, a stretch bound $(m,a)$, and a bound on the number of outliers $f\in\mathbb{N}$.  The goal is to find a string $r$ as well as a set $F\subseteq V$ where $|F|\le f$, in order to minimize $|preprocess(V\backslash F,d,m,a,r)|$.
\end{definition}

We will provide both true approximation results and ($\alpha,\beta$)-bicriteria results, in which we slightly violate the bound on the number of outliers.   Formally, an $(\alpha, \beta)$-approximation algorithm for the $\mathcal{A}$-optimization problem with outliers is on algorithm which on any input $((V,d), (m,a), f)$ returns a solution with cost at most $\alpha\cdot OPT$ that has at most $\beta\cdot f$ outliers (where $OPT$ is the minimum cost of any solution with at most $f$ outliers). 

\subsection{Our Results and Techniques}

With these definitions in hand, we can now formally state our results.

In Section~\ref{sec:tz2} we discuss the problem of optimizing the $3$-stretch Thorup-Zwick distance oracle, i.e., the $TZ_2$-optimization problem. It is straightforward to obtain an $O(\log n)$-approximation by reducing to the non-metric facility location problem.

\begin{theorem}\label{thm:utz2}
There is an $O(\log n)$-approximation algorithm for the $TZ_2$-optimization problem.
\end{theorem}

To prove a matching lower bound, we use a reduction from Label Cover to the $TZ_2$-optimization problem. We use a proof which is similar to the proof of the hardness of Set Cover in~\cite{Vazirani13} (based on~\cite{Feige98}). However, we cannot use a reduction directly from Set Cover since we will need some extra properties of the starting instances, and thus are forced to start from Label Cover.  We introduce a new notion of $(m, l, \delta)$-set families and show that these can still be plugged into existing hardness results to get the extra structural properties that we need.   This lets us prove the following theorem:

\begin{theorem}\label{thm:ltz2}
Unless $\mathbf{NP}\subseteq\mathbf{DTIME}(n^{O(\log\log n)})$, the $TZ_2$-optimization problem does not admit a polynomial-time $o(\log n)$-approximation.  
\end{theorem}

For larger stretch values, a natural approach is to realize that a simple LP relaxation suffices to give Theorem~\ref{thm:utz2} in the stretch $3$ case, and try to extend this basic LP to larger stretches.  In Section~\ref{sec:tzk}, we show that this does not work for the more general $TZ_k$-optimization problem: the integrality gap jumps up to become a polynomial. The instance is very simple: it is just the metric space formed by shortest paths on the $n$-cycle.  It turns out to be straightforward to calculate the optimal fractional LP cost, but proving that the optimal integral solution is large is surprisingly involved.  

\begin{theorem}\label{thm:gaptzk}
The basic LP relaxation for the $TZ_k$-optimization problem has an $\Omega(n^{\frac{1}{k}-\frac{1}{2^{k-1}}})$ integrality gap when $k>2$.
\end{theorem}

In Section~\ref{sec:pr} we discuss the problem of optimizing the P\v{a}tra\c{s}cu-Roditty distance oracle. The basic LP and a simple rounding algorithm gives us an $O(\log n)$-approximation algorithm.

\begin{theorem}\label{thm:upr}
There is an $O(\log n)$-approximation algorithm for $PR$-optimization problem.
\end{theorem}

A reduction from set cover problem also gives us a matching lower bound.

\begin{theorem}\label{thm:lpr}
Unless $\mathbf{P}=\mathbf{NP}$, the $PR$-optimization problem does not admit a polynomial-time $o(\log n)$-approximation.  
\end{theorem}

In Section~\ref{sec:doo} we move to the outliers setting.  For both $TZ_2$- and $PR$-optimization problems, a semidefinite programming relaxation and a simple rounding algorithm gives us an $(O(\frac{\log n}{\varepsilon}),1+\varepsilon)$-approximation algorithm.  Using an SDP relaxation seems to be necessary -- the corresponding LP relaxation requires violating the number of outliers by a factor of $2$ rather than a factor of $1+\varepsilon$.  We can also get a true approximation on $TZ_2$-optimization problem with outliers if the number of outliers is low.  These results form the following theorems.

\begin{theorem}\label{thm:tz2o}
There is an $(O(\frac{\log n}{\varepsilon}),1+\varepsilon)$-approximation algorithm for the $TZ_2$-optimization problem with outliers.
\end{theorem}

\begin{theorem}\label{thm:tz2ot}
There is an $O(\log n)$-approximation algorithm for $TZ_2$-optimization problem with outliers if the number of outliers is at most $\sqrt{n}$.
\end{theorem}

\begin{theorem}\label{thm:pro}
There is an $(O(\frac{\log n}{\varepsilon}),1+\varepsilon)$-approximation algorithm for the $PR$-optimization problem with outliers.
\end{theorem}

\section{$TZ_2$-Optimization Problem}\label{sec:tz2}

We first give an $O(\log n)$-approximation for $TZ_2$-optimization (Theorem~\ref{thm:utz2}), and follow this with a matching lower bound.  

\subsection{Upper Bound}\label{sec:utz2}

We will prove our upper bound by a reduction to the non-metric facility location problem.

\begin{definition}
In the \emph{non-metric facility location} problem we are given a set $F$ of facilities, a set $D$ of clients, an opening cost function $f:F\rightarrow\mathbb{R}^+$, and a connection cost function $c:D\times F\rightarrow\mathbb{R}^+$. The goal is to find the set $S \subseteq F$ which minimizes $\sum_{i\in S}f(i)+\sum_{i\in D}\min_{j\in S}c(i,j)$ (i.e.~the sum of the opening and connection costs).
\end{definition}

Non-metric facility location is a classic problem, and much is known about it, including the following upper bound due to Hochbaum. 

\begin{theorem}[\cite{Hochbaum82}] \label{thm:Hochbaum}
There is a polynomial time algorithm which gives an $O(\log n)$-approximation to the non-metric facility location problem.
\end{theorem}

Hochbaum's algorithm is a greedy algorithm, but it is also straightforward to design an algorithm with similar bounds using an LP relaxation.  Since it is not necessary we do not present the relaxation here, but generalizations of the relaxation will prove important in the more general $TZ_k$ setting (see Section~\ref{sec:tzk}).

We now show that the  $TZ_2$-optimization problem is essentially a special case of non-metric facility location problem.  First, simple arithmetic manipulation of the cost function of the $TZ_2$-optimization problem gives the following:
\begin{align*}
cost(A_1,V,d)=&\sum_{u\in V}|R_{1u}|+\sum_{u\in V}|R_{2u}| \\
=&\sum_{u\in V}\left|\{v\in V\mid d(u,v)<\min_{w\in A_1}d(u,w)\}\right| +\sum_{u\in V}\left|\{v\in A_1\mid d(u,v)<\infty\right|\\
=&\sum_{u\in V}\left|\{v\in V\mid d(u,v)<\min_{w\in A_1}d(u,w)\}\right|+n|A_1| \\
=&\sum_{w\in A_1}n+\sum_{u\in V}\min_{w\in A_1}\left|\{v\in V\mid d(u,v)<d(u,w)\}\right|.
\end{align*}

Given an instance $(V,d)$ of the $TZ_2$-optimization problem, we create an instance of non-metric facility location by setting $F=D=V$, opening costs $f(v)=n$ for all $v\in V$, and connection costs $c(u,w)=|\{v\in V\mid d(u,v)<d(u,w)\}|$ for all $u,w\in V$.  Then the cost function of the $TZ_2$-optimization problem is exactly the same as the cost function of non-metric facility location problem. Therefore $TZ_2$ is a special case of non-metric facility location, which together with Theorem~\ref{thm:Hochbaum} implies Theorem~\ref{thm:utz2}.

\subsection{Lower Bound}

Proving an $\Omega(\log n)$ hardness of approximation (Theorem~\ref{thm:ltz2}) turns out to be surprisingly difficult.  Details appear in Appendix~\ref{app:tz2}; here we provide an informal overview.  Technically we reduce directly to $TZ_2$-optimization from a version of the Label Cover problem that corresponds to applying parallel repetition~\cite{RR98} to 3SAT-5, which is a standard starting point for hardness reductions.  Informally, though, we are ``really" reducing from Set Cover: given an instance of Set Cover, we show how to create an instance of $TZ_2$-optimization where the cost of the optimal solution is the same (up to a constant and a polynomial scaling factor).  But in order for our reduction to work, we actually need more than just an arbitrary Set Cover instance: we need a version of Set Cover in which it is hard even to cover \emph{most} of the elements, not just all of them.  

So we have to also give a new reduction from Label Cover to Set Cover, showing that even this version of Set Cover is hard.  It turns out that Feige's reduction~\cite{Feige98}, reinterpreted by Vazirani~\cite{Vazirani13}, essentially already gives us what we need.  We just need to analyze it a bit more carefully.  In particular, a key component of this reduction is what Vazirani called \emph{$(m, l)$-set systems}, which can be thought of as nearly-unbiased sample spaces.  We generalize this notion to $(m, l, \delta)$-set systems, given in the following definition.

%For the lower bound, what we did is essentially a two step reduction. The first step is a reduction from Label Cover problem, which is a common starting point for hardness results, to the Set Cover problem. The second step is a reduction from Set Cover problem to the $TZ_2$-optimization problem. However, our reduction on the second step has to have some restrictions on the set cover instance. For example, the number of sets should be basically the same as the number of elements. We also requires that it is hard not only for finding a good solution which covers all the elements, but also hard for finding a good solution which covers most of the elements.

%The first requirement can be easily satisfied by duplicating sets, but for the second requirement, we used a notion of ($m,l,\delta$)-set system. A ($m,l,\delta$)-set system is a nearly random sample space, and it is an extension of the concept of ($m,l$)-set system (which is essentially a ($m,l,0$)-set system) written in~\cite{Vazirani13}, which is used to prove Feige's~\cite{Feige98} set cover inapproximability result. 

\begin{definition}~\label{def:setsystem}
A set $B$ (the universe) and a collection of subsets $C_1,\mathellipsis,C_m$ of $B$ form an \emph{($m,l,\delta$)-set system} if any collection of $l$ sets in $\{C_1,\mathellipsis,C_m,\overline{C_1},\mathellipsis,\overline{C_m}\}$ whose union contains at least $(1-\delta)|B|$ elements must include both $C_i$ and $\overline{C_i}$ for some $i$.
\end{definition}

An $(m, l)$-set system is just a $(m, l, 0)$-set system.  While not all $(m, l)$-set systems are $(m, l \delta)$-set systems for larger $\delta$, the construction of $(m, l)$-set systems in~\cite{Vazirani13} actually does generalize directly to larger values of $\delta$.  With this tool in hand, we follow through the rest of the reduction and it gives us the type of Set Cover instances which we need.  Technically our reduction skips this step by going directly from Label Cover to $TZ_2$-optimization, but generating these kinds of Set Cover instances is intuitively what the first part of the reduction is doing.  

%, but with more careful discussion. Due to the page limit, we will move the formal definition of the Label Cover problem, the construction of ($m,l,\delta$)-set system, the reduction, and the proofs to appendix~\ref{app:tz2}. In addition, for the simplicity of the proof, we just do a reduction directly from Label Cover problem to $TZ_2$-optimization problem.
%
%Consider the Label Cover problem as a \emph{decision problem} with input $(G=(V',E),\Sigma,\Pi)$ where $G$ is a graph, $\Sigma$ is a set, and $\Pi=\{\Pi_e:\Sigma\rightarrow\Sigma\mid e\in E\}$ is a set of function. We will prove in the appendix that there exists a polynomial reduction $(G,\Sigma,\Pi)\rightarrow(V,d)$ such that:
%\begin{itemize}
%\item If ($G=(V',E),\Sigma,\Pi$) is a YES instance in the Label Cover problem. Then the reduction $(V,d)$ to the $TZ_2$-optimization problem has a solution with cost $\le(|V'|+1)\cdot|V|$.
%\item If ($G=(V',E),\Sigma,\Pi$) is a No instance in the Label Cover problem. Then the reduction $(V,d)$ to the $TZ_2$-optimization problem has no solution with cost $<\Theta(\log|V|)\cdot|V'|\cdot|V|$.
%\end{itemize}
%which together with the hardness of label cover problem proved Theorem~\ref{thm:ltz2}

\section{$TZ_k$-Optimization Problem}\label{sec:tzk}

We now move to the more general $TZ_k$-optimization problem. While we are not able to give nontrivial upper bounds for this problem, we can at least show that the basic LP relaxation (as discussed in Section~\ref{sec:utz2}) does not give polylogarithmic bounds in this more general setting.

%\mdnote{Need to give an overview here of what we're doing in this section.  ``We now move to the more general $TZ_k$-optimization problem.  While we are not able to give nontrivial upper bounds for this problem, we can at least show that the basic LP relaxation (which suffices to give an $O(\log n)$-approximation for $TZ_2$-optimization) does not give polylogarithmic bounds in this more general setting}
%\subsection{Probem Definition}
%
%Given a metric space $(V,d)$ where $|V|=n$, we want to find a sequence of subsets $\varnothing=A_k\subseteq A_{k-1}\subseteq\mathellipsis\subseteq A_0=V$ which minimizes
%\[cost(A_0,\mathellipsis,A_k,V,d)=\sum_{u\in V}\sum_{i=1}^k\left|\{v\in A_{i-1}\backslash\{u\}\mid d(u,v)<\min_{w\in A_i}d(u,w)\}\right|.\]
%We define $\min_{u\in\varnothing}d(u,v)=\infty$ for every $v\in V$. We also say that a vertex $u$ needs to remember the distance to another vertex $v$ if $v\in A_{i-1}\backslash\{u\}$ and $d(u,v)<\min_{w\in A_i}d(u,w)\}$
%
%Thorup-Zwick shows that for every metric space $(V,d)$ and a constant $k$, there is an algorithm to find these subsets to make sure that the total cost is $O(kn^{1+\frac{1}{k}})$. Note that the minimun cost of this kind of distance oracle is at least $n$, Thorup-Zwick actually gives us an $O(n^\frac{1}{k}))$-approximation. The problem is: Can we have a better approximation?

\subsection{The LP}

Let $B_u(v)=\{w\in V\mid d(u,w)\le d(u,v)\}$. For every $v\in V$ and $i\in [k]$, let $x_v^{(i)}$ be a variable which is supposed to be an indicator for whether $v\in A_i$.  Similarly, for all $u,v\in V$ and $i\in [k]$, let $y_{uv}^{(i)}$ be a variable which is supposed to be an indicator for whether $v\in R_{iu}$. (Recall that $R_{iu}=\{v\in A_{i-1}\mid d(u,v)<\min_{w\in A_i}d(u,w)\}$) We can easily write an LP relaxation for this problem:
\[\begin{array}{rll}(LP_{TZ_k}):
\min&\sum_{i=1}^k\sum_{u,v\in V}y_{uv}^{(i)}\\
s.t.&0=x_v^{(k)}\le x_v^{(k-1)}\le\mathellipsis\le x_v^{(1)}\le x_v^{(0)}=1&\forall v\in V\\
&y_{uv}^{(i)}\ge x_v^{(i-1)}-\sum_{w\in B_u(v)}x_w^{(i)}&\forall u,v\in V,i\in[k]\\
&y_{uv}^{(i)}\ge0&\forall u,v\in V,i\in[k]\\
\end{array}\]

It can easily be shown that this is a valid relaxation (the proof can be found in Appendix~\ref{app:tzk}).  When restricted to the special case of $k=2$, it is not hard to see that this LP is essentially a special case of the basic LP relaxation for non-metric facility location, which can be used to prove the $O(\log n)$ bound of Theorem~\ref{thm:utz2}.  But for larger values of $k$ the behavior is different, and does not result in a polylogarithmic integrality gap.

\subsection{Integrality Gap}

The integrality gap instance is quite simple: the metric $(V,d)$ induced by shortest-path distances in a cycle.  Slightly more formally, we let $V=[n]$, and use the cycle distance $d(u,v)=\min\{|u-v|,n+\min\{u,v\}-\max\{u,v\}\}$. 

Details can be found in Appendix~\ref{app:tzk}. It turns out to be relatively easy to find a fractional solution to $LP_{TZ_k}$ with cost $O(n^{1+\frac{1}{2^{k-1}}})$ on this instance.  The tricky part is lower bounding the optimal solution, i.e., showing that the optimal integral solution has cost at least $\Omega(n^{1+\frac{1}{k}})$. Combining these two results gives us an $\Omega(n^{\frac{1}{k}-\frac{1}{2^k-1}})$ integrality gap, proving Theorem~\ref{thm:gaptzk}.

\section{$PR$-Optimization Problem}\label{sec:pr}

We now move from Thorup-Zwick distance oracles to P\v{a}tra\c{s}cu-Roditty distance oracles.  We show that from an optimization perspective, they are similar to $TZ_2$ in that we can find matching bounds: an $O(\log n)$-approximation, and $\Omega(\log n)$-hardness.

\subsection{Upper Bound}

In this section we prove Theorem~\ref{thm:upr} by using an LP and randomized rounding to give an $O(\log n)$-approximation to the $PR$-optimization problem.

Let $B_u(v)=\{w\in V\mid d(u,w)\le d(u,v)\}$, and $B(u,r)=\{w\in V\mid d(u,w)\le r\}$. We can see $B_u(v)=B(u,d(u,v))$. Now, let $x_v$ be a variable which is supposed to be an indicator for whether $v\in A$, and let $y_{uv}$ be a variable which is supposed to be an indicator for whether $\{u,v\}\in R$. (Recall that $R=\left\{\{u,v\}\subseteq V\mid d(u,v)<\min_{w\in A}d(u,w)+\min_{w\in A}d(v,w)-1\right\}$). We can write the following LP relaxation:

\[\begin{array}{rll}(LP_{PR}):
\min&\sum_{v\in V}n\cdot x_v+\sum_{\{u,v\}\subseteq V}y_{uv}\\
s.t.&y_{uv}\ge1-\sum_{w\in B_u(r)\cup B_v(d(u,v)-r)}x_w&\forall u,v\in V,\forall r\in[0,d(u,v)]\\
&x_v\in[0,1]&\forall v\in V\\
&y_{uv}\ge0&\forall u,v\in V\\
\end{array}\]

At first blush it may not be obvious that the first type of constraint in this LP really captures the characterization of paris in $R$.  But it is actually not that hard to see that this is a valid relaxation (a formal proof can be found in Appendix~\ref{app:pr}). Note that while the number of constraints appears to be exponential (recall that we assume integer weights, but not necessarily unit weights, and hence $d(u,v)$ is not necessarily polynomial in the input size), it is in fact possible to solve this LP in polynomial time.  We can do this by noting that for each $u, v \in V$, only at most $n$ different value of $r$ actually yield \emph{different} constraints, so we can simply write the constraints for those values. 

Our algorithm is relatively straightforward. We first solve $LP_{PR}$ and get an optimal fractional solution $(x_v^*,y_{uv}^*)$.  We then use independent randomized rounding, adding each $v \in V$ to $A$ independently with probability $\min\{4\ln n\cdot x_v^*,1\}$.

\begin{lemma}
If $y_{uv}^*\le\frac{1}{2}$, then the probability that $\{u,v\}\in R$ is at most $\frac{1}{n}$.
\end{lemma}
\begin{proof}
If $y_{uv}^*\le\frac{1}{2}$, then the first constraint implies that $\sum_{w\in B_u(r)\cup B_v(d(u,v)-r)}x_w^*\ge\frac{1}{2}$ for all $r\in[0,d(u,v)]$.  Therefore, the probability that $A\cap(B_u(r)\cup B_v(d(u,v)-r))=\varnothing$ for a specific $r\in[0,d(u,v)]$ is at most
\[\prod_{w\in B_u(r)\cup B_v(d(u,v)-r)}(1-\min\{4\ln n\cdot x_w^*,1\})\le e^{-\sum_{w\in B_u(r)\cup B_v(d(u,v)-r)}4\ln n\cdot x_w^*}\le\frac{1}{n^2}\]

A union bound over all the different values of $r$ we used in our LP implies that  the probability that there exists an $r\in[0,d(u,v)]$ where $A\cap(B_u(r)\cup B_v(d(u,v)-r))=\varnothing$ is at most $\frac{1}{n^2}\cdot n=\frac{1}{n}$.  We claim that the existence of such an $r$ is implied by $\{u,v\} \in R$, and hence the probability that $\{u,v\} \in R$ is at most $\frac1n$.  To see this, suppose that $\{u,v\} \in R$, i.e.~suppose that  $d(u,v)<\min_{w\in A}d(u,w)+\min_{w\in A}d(v,w)-1$.  Then if we set $r = \min_{w \in A} d(u,w) - 1$, this implies that $\min_{w \in A} d(v,w) > d(u,v) - r$.  But then this would imply that no element of $A$ is in $B_u(r)\cup B_v(d(u,v)-r)$.
%
%Thus with probability at least $1-\frac1n$, for every $r$ there is some element of $A$ contained in $B_u(r)\cup B_v(d(u,v)-r)$.  
%
%We claim that this implies that $\{u,v\} \not\in R$.  To see this, consider the contrapositive: suppose that $\{u,v\} \in R$, i.e.~suppose that  $d(u,v)<\min_{w\in A}d(u,w)+\min_{w\in A}d(v,w)-1$.  Then if we set $r = \min_{w \in A} d(u,w) - 1$, this implies that $\min_{w \in A} d(v,w) > d(u,v) - r$.  But then this would imply that no element of $A$ is in $B_u(r)\cup B_v(d(u,v)-r)$.
%
%Thus the probability that $\{u,v\} \not\in R$ is at least $1-\frac1n$, implying the lemma.
%If there is some node $w'$ in $A\cap(B_u(r)\cup B_v(d(u,v)-r))$, then $\min_{w\in A}d(u,w)+\min_{w\in A}d(v,w) \leq r + (d(u,v) - r) = d(u,v)$, and hence $\{u,v\} \not\in R$ \mdnote{Seems like an off by one error?}.  Thus the probability that $\{u,v\}\in R$ is at most $\frac{1}{n}$.
\end{proof}

Let $OPT_{LP_{PR}}$ denote the optimal cost of $LP_{PR}$.  Then the above lemma implies that the expected cost of the rounding algorithm is at most
\begin{align*}
\textbf{E}[n|A| + |R|] &\leq \sum_{v\in V}n\cdot x_v^*\cdot4\ln n+2\cdot\sum_{u,v\in V}y_{uv}^*+n^2\cdot\frac{1}{n}\le O(\log n)\cdot OPT_{LP_{PR}}+n\\ 
&\le O(\log n)\cdot OPT
\end{align*}
(where we use the fact that $OPT\ge\Omega(n)$).  This completes the proof of Theorem~\ref{thm:upr}.

\subsection{$\Omega(\log n)$-hardness}

We now show a matching hardness bound for the $PR$-optimization problem by reducing from the Set Cover problem.

Consider a set cover instance $(\mathcal{U},\mathcal{S})$ where $|\mathcal{U}|+|\mathcal{S}|=n$. For each $e\in\mathcal{U}$, we create a group of vertices $G_e$ where $|G_e|=3n$. For each $S\in\mathcal{S}$, we also create a group of vertices $G_S$ where $|G_S|=3n$.

Now we construct the following metric space: $V=(\bigcup_{e\in\mathcal{U}}G_e)\cup(\bigcup_{S\in\mathcal{S}}G_S)$ and
\[d(u,v)=
\begin{cases} 
1,&\mbox{if }u\in G_e,v\in G_e\\
1,&\mbox{if }u\in G_S,v\in G_S\\
1,&\mbox{if }u\in G_e,v\in G_S,e\in S\\
2,&\mbox{otherwise.}
\end{cases}\]

In Appendix~\ref{app:pr} we show that if there is a solution $\mathcal{S}^*$ to the set cover instance $(\mathcal{U},\mathcal{S})$ where $|\mathcal{S}^*|=t$, then there is a set $A$ where $cost(A,V,d)\le t|V|$.  We also show that if there is a set $A \subseteq V$ where $cost(A,V,d)\le t|V|$, then there exists a solution $\mathcal{S^*}$ to the set cover instance $(\mathcal{U},\mathcal{S})$ where $|\mathcal{S}^*|=t$.  These two claims, together with an appropriate hardness theorem for Set Cover~\cite{DS14}, imply Theorem~\ref{thm:lpr}.

%Since there is a $\Omega(\log n)$-hardness result on approximating set cover problem, Theorem~\ref{thm:lpr} is obvious from the reduction.

\section{Distance Oracles With Outliers}\label{sec:doo}

We now move to the more difficult outliers setting, where we can also optimize over a set of vertices to ignore.  Recall that for an approximate distance oracle $\mathcal A$, our goal is now to find a set of vertices $F$ (the outliers) where $|F| \leq f$ as well as a string $r$ in order to minimize $|preprocess(V \setminus F,d,m,a,r)|$.  In other words, we are going to try to solve the same problems as before, but where we can choose a set $F$ to remove.  We begin with $TZ_2$, and then move to $PR$.

\subsection{$TZ_2$-Optimization Problem With Outliers}

For this problem, it is easy to see that the cost function becomes:
\[cost(A,F,V,d)=(n-f)|A|+\sum_{u\in V\backslash F}|R_{1u}|=(n-f)|A|+\sum_{u\in V\backslash F}\left|\{v\in V\backslash F\mid d(u,v)<\min_{w\in A}d(u,w)\}\right|\]

%\subsubsection{Probem Definition}
%
%Given a metric space $(V,d)$ where $|V|=n$ and a number $f\in\mathbb{N}$, we want to find a set $A\subseteq V$ and another set $F\subseteq V$ where $A\cap F=\varnothing$, $|F|=f$ and minimizes
%\[cost(A,F,V,d)=(n-f-1)|A|+\sum_{u\in V\backslash F}\left|\{v\in V\backslash(F\cup\{u\})\mid d(u,v)<\min_{w\in A}d(u,w)\}\right|\]
%
%\begin{definition}{($\alpha,\beta$)-approximation to TZ with outliers:}
%Given $V,d$ and $f$, assume the optimal solution has cost $OPT$. An ($\alpha,\beta$)-approximation to TZ with outliers is two sets $A,F$ where $cost(A,F,V,d)\le\alpha OPT$ and $|F|\le\beta f$
%\end{definition}

%\subsubsection{LP}

A natural approach is to use an LP which is similar to $LP_{TZ_k}$ to solve this problem (but for $TZ_2$), suitably adapted to handle outliers. Let $x_v$ be a variable which is supposed to be an indicator for whether $v\in A$, let $y_{uv}$ be a variable which is supposed to be an indicator for whether $v\in R_{1u}$, and let $z_v$ be a variable which is supposed to be an indicator for whether $v\in F$. Then we can write the following natural LP relaxation: 
\[\begin{array}{rll}(LP_{TZ_2O}):
\min&\sum_{v\in V}(n-f)\cdot x_v+\sum_{u,v\in V}y_{uv}\\
s.t.&y_{uv}\ge1-z_u-z_v-\sum_{w\in B_u(v)}x_w&\forall u,v\in V\\
&\sum_{v\in V}z_v\le f\\
&x_v\in[0,1]&\forall v\in V\\
&y_{uv}\ge0&\forall u,v\in V\\
&z_v\in[0,1]&\forall v\in V\\
\end{array}\]

%The first constraint means $u$ needs to remember $v$ if neither of $u,v$ is outlier, and $A\cap B_u(v)=\varnothing$. The second constraint means the total number of outliers is at most $f$.

Unfortunately, this LP can not give an $(\alpha,\beta)$-approximation with $\beta=2-\epsilon$. To see this, consider the case that $f=\frac{n}{2}$.  Then the optimal solution to $LP_{TZ_2O}$ is $0$, by setting all $z_v$ to $\frac{1}{2}$, all $x_v$ to $0$, and all $y_{uv}=0$.  Thus any integral solution, to be competitive with this fractional solution, must treat \emph{all} nodes as outliers, requiring $\beta$ to be at least $2$.  

%\subsubsection{SDP}

Fortunately we can give a stronger semidefinite programming relaxation, allowing for a better approximation. As in $LP_{TZ_2O}$, let $\vec{x}_v$ be a variable which is supposed to be an indicator for whether $v\in A$, let $\vec{y}_{uv}$ be a variable which is supposed to be an indicator for whether $v\in R_{1u}$, and let $\vec{z}_v$ be a variable which is supposed to be an indicator for whether $v\in F$. We can then write this SDP:
\[\begin{array}{rll}(SDP_{TZ_2O}):
\min&\sum_{v\in V}(n-f)\cdot\|\vec{x}_v\|^2+\sum_{u,v\in V}\|\vec{y}_{uv}\|^2\\
s.t.&\|\vec{y}_{uv}\|^2\ge1-\vec{z}_u\cdot\vec{z}_v-\sum_{w\in B_u(v)}\|\vec{x}_w\|^2&\forall u,v\in V\\
&\sum_{v\in V}\|\vec{z_v}\|^2\le f\\
&\|\vec{x}_v\|^2\le1&\forall v\in V\\
&\|\vec{y}_{uv}\|^2\le1&\forall u,v\in V\\
&\|\vec{z}_v\|^2\le1&\forall v\in V\\
\end{array}\]

Our approximation algorithm first solves $SDP_{TZ_2O}$ to get an optimal solution $(\vec{x}_v^*,\vec{y}_{uv}^*,\vec{z}_v^*)$.  We then use independent randomized rounding to construct $A$, adding each $v \in V$ to $A$ independently with probability $\min\{\frac{3\ln n}{\varepsilon}\cdot\|\vec{x}_v^*\|^2,1\}$ where $\varepsilon$ is a small constant.  Finally, we use threshold rounding to construct $F$ by adding each $v\in V$ to $F$ if $\|\vec{z}_v^*\|^2\ge\frac{1}{1+\varepsilon}$.

We want to show that this is an ($O(\log n),1+\varepsilon$)-approximation. It is easy to see that $|F|\le(1+\varepsilon)f$ because $\sum_{v\in V}\|\vec{z}_v^*\|^2\le f$. In order to prove Theorem~\ref{thm:tz2o} it only remains to prove that the expected cost is at most $O(\log n)\cdot OPT$.  This proof can be found in Appendix~\ref{app:doo}.

When $f \leq \sqrt{n}$ we can actually give a true $O(\log n)$-approximation (Theorem~\ref{thm:tz2ot}).  The algorithm is almost the same; we just need to change the threshold rounding for outliers to instead pick the $f$ vertices with largest $\|\vec{z}_v\|^2$ value. Details appear in Appendix~\ref{app:doo}.

\subsection{$PR$-Optimization Problem With Outliers}

%\subsubsection{Probem Definition}
%
%Given a metric space $(V,d)$ where $|V|=n$ and a number $f\in\mathbb{N}$, we want to find a $A\subseteq V$ and another set $F\subseteq V$ where $A\cap F=\varnothing$, $|F|=f$ and minimizes

For this problem, the cost function becomes:
\begin{align*}
cost(A,F,V,d)=&(n-f)\cdot|A|+|R|\\
=&(n-f)\cdot|A|+\left|\{\{u,v\}\subseteq V\backslash F\mid d(u,v)<\min_{w\in A}d(u,w)+\min_{w\in A}d(v,w)-1\}\right|.
\end{align*}

%We have a similiar SDP and ($O(\log n),1+\varepsilon$)-approximation for this problem.

We will again use an SDP relaxation.  Let $\vec{x}_v$ be a variable which is supposed to be an indicator for whether $v\in A$, let $\vec{y}_{uv}$ be a variable which is supposed to be an indicator for whether $\{u,v\}\in R$, and let $\vec{z}_v$ be a variable which is supposed to be an indicator for whether $v\in F$. We have the following relaxation which is similar to both $LP_{PR}$ and $SDP_{TZ_2O}$:
\[\begin{array}{rll}(SDP_{PR}):
\min&\sum_{v\in V}(n-f)\cdot\|\vec{x}_v\|^2+\sum_{\{u,v\}\subseteq V}\|\vec{y}_{uv}\|^2\\
s.t.&\|\vec{y}_{uv}\|^2\ge1-\vec{z}_u\cdot\vec{z}_v-\sum_{w\in B_u(r)\cup B_v(d(u,v)-r)}\|\vec{x}_w\|^2&\forall u,v\in V,r\in[0,d(u,v)]\\
&\sum_{v\in V}\|\vec{z_v}\|^2\le f\\
&\|\vec{x}_v\|^2\le1&\forall v\in V\\
&\|\vec{y}_{uv}\|^2\le1&\forall u,v\in V\\
&\|\vec{z}_v\|^2\le1&\forall v\in V\\
\end{array}\]

Note that this $SDP_{PR}$ is solvable in polynimial time for the same reason that $LP_{PR}$ is solvable: for each pair of $(u,v)$, we can find $n$ different values of $r$ that give all of the distinct constraints.

The rounding algorithm is basically the same as the $TZ_2$-optimization problem with outliers. We first solve the $SDP_{PR}$ and get an optimal solution $(\vec{x}_v^*,\vec{y}_{uv}^*,\vec{z}_v^*)$. We then use independent randomized rounding to get $A$, adding each $v\in V$ to $A$ independently with probability $\min\{\frac{6\ln n}{\varepsilon}\cdot\|\vec{x}_v^*\|^2,1\}$ where $\varepsilon$ is a small constant.  Then we use threshold rounding to get $F$, adding each $v\in V$ to $F$ if $\|\vec{z}_v^*\|^2\ge\frac{1}{1+\varepsilon}$.

This is an ($O(\log n),1+\varepsilon$)-approximation. It is easy to see that $|F|\le(1+\varepsilon)f$ because $\sum_{v\in V}\|\vec{z}_v^*\|^2\le f$. The proof that the expected cost is at most $O(\log n)\cdot OPT$ is in Appendix~\ref{app:doo}, which completes the proof of Theorem~\ref{thm:pro}.

\section{Conclusion and Future Work}

In this paper we initiate the study of \emph{approximating} approximate distance oracles.  This is a different take on the question of optimizing data structures, where we attempt to find the best data structure for a particular input, rather than for a class of inputs.  In order to make this tractable (or even well-defined), we restrict our attention to known classes of distance oracles, and show that it is sometimes possible to find the best of these restricted oracles.  We also extended our approaches to optimize in the presence of outliers.

For future work, the major question is clearly whether we can approximately optimize higher level (i.e., higher stretch) Thorup-Zwick distance oracles.  Although we show an integrality gap for the basic LP, it is quite conceivable that a stronger LP or SDP could be used to give a logarithmic approximation ratio. Beyond this, there are other distance oracles which could be optimized -- we chose Thorup-Zwick and P\v{a}tra\c{s}cu-Roditty since they are well-known and in some ways canonical, but it would be interesting to extend these ideas to other oracles.  At a higher level, we believe that the definitions and ideas we have introduced here could lead to many interesting questions about optimizing data structures for given inputs: can we find near-optimal distance labels?  Or compact routing schemes?  Or connectivity oracles? Or fault-tolerant oracles?  Essentially any data structure question in which there is a choice of \emph{which} data to store, rather than how to store it, can be put into our optimization framework.  Exploring this space is an exciting future direction.  

\clearpage

\bibliographystyle{plain}
\bibliography{TZ}

\begin{thebibliography}{10}

\bibitem{AG11}
Ittai Abraham and Cyril Gavoille.
\newblock On approximate distance labels and routing schemes with affine
  stretch.
\newblock In {\em International Symposium on Distributed Computing}, pages
  404--415. Springer, 2011.

\bibitem{AOJR92}
Noga Alon, Oded Goldreich, Johan H{\aa}stad, and Ren{\'e} Peralta.
\newblock Simple constructions of almost k-wise independent random variables.
\newblock {\em Random Structures \& Algorithms}, 3(3):289--304, 1992.

\bibitem{BGKW88}
Michael Ben-Or, Shafi Goldwasser, Joe Kilian, and Avi Wigderson.
\newblock Multi-prover interactive proofs: How to remove intractability
  assumptions.
\newblock In {\em Proceedings of the twentieth annual ACM symposium on Theory
  of computing}, pages 113--131. ACM, 1988.

\bibitem{BBMRY13}
Piotr Berman, Arnab Bhattacharyya, Konstantin Makarychev, Sofya Raskhodnikova,
  and Grigory Yaroslavtsev.
\newblock Approximation algorithms for spanner problems and directed steiner
  forest.
\newblock {\em Information and Computation}, 222:93--107, 2013.

\bibitem{CDGKS09}
T-H~Hubert Chan, Kedar Dhamdhere, Anupam Gupta, Jon Kleinberg, and Aleksandrs
  Slivkins.
\newblock Metric embeddings with relaxed guarantees.
\newblock {\em SIAM Journal on Computing}, 38(6):2303--2329, 2009.

\bibitem{CDG06}
T-H~Hubert Chan, Michael Dinitz, and Anupam Gupta.
\newblock Spanners with slack.
\newblock In {\em European Symposium on Algorithms}, pages 196--207. Springer,
  2006.

\bibitem{Chechik14}
Shiri Chechik.
\newblock Approximate distance oracles with constant query time.
\newblock In {\em Proceedings of the 46th Annual ACM Symposium on Theory of
  Computing}, pages 654--663. ACM, 2014.

\bibitem{Chechik15}
Shiri Chechik.
\newblock Approximate distance oracles with improved bounds.
\newblock In {\em Proceedings of the Forty-Seventh Annual ACM on Symposium on
  Theory of Computing}, pages 1--10. ACM, 2015.

\bibitem{CD14}
Eden Chlamt{\'{a}}c and Michael Dinitz.
\newblock Lowest degree k-spanner: Approximation and hardness.
\newblock In {\em {APPROX-RANDOM}}, volume~28 of {\em LIPIcs}, pages 80--95.
  Schloss Dagstuhl - Leibniz-Zentrum fuer Informatik, 2014.

\bibitem{CDK12}
Eden Chlamtac, Michael Dinitz, and Robert Krauthgamer.
\newblock Everywhere-sparse spanners via dense subgraphs.
\newblock In {\em 53rd Annual {IEEE} Symposium on Foundations of Computer
  Science, {FOCS} 2012, New Brunswick, NJ, USA, October 20-23, 2012}, pages
  758--767. {IEEE} Computer Society, 2012.

\bibitem{DK11}
Michael Dinitz and Robert Krauthgamer.
\newblock Directed spanners via flow-based linear programs.
\newblock In {\em Proceedings of the forty-third annual ACM Symposium on Theory
  of computing}, pages 323--332. ACM, 2011.

\bibitem{DZ16}
Michael Dinitz and Zeyu Zhang.
\newblock Approximating low-stretch spanners.
\newblock In {\em Proceedings of the Twenty-Seventh Annual ACM-SIAM Symposium
  on Discrete Algorithms}, pages 821--840. SIAM, 2016.

\bibitem{DS14}
Irit Dinur and David Steurer.
\newblock Analytical approach to parallel repetition.
\newblock In {\em Proceedings of the 46th Annual ACM Symposium on Theory of
  Computing}, pages 624--633. ACM, 2014.

\bibitem{Feige98}
Uriel Feige.
\newblock A threshold of ln n for approximating set cover.
\newblock {\em Journal of the ACM (JACM)}, 45(4):634--652, 1998.

\bibitem{Hochbaum82}
Dorit~S Hochbaum.
\newblock Heuristics for the fixed cost median problem.
\newblock {\em Mathematical programming}, 22(1):148--162, 1982.

\bibitem{Jacobson88}
Guy~Joseph Jacobson.
\newblock Succinct static data structures.
\newblock 1988.

\bibitem{PR10}
Mihai Patrascu and Liam Roditty.
\newblock Distance oracles beyond the thorup-zwick bound.
\newblock In {\em Foundations of Computer Science (FOCS), 2010 51st Annual IEEE
  Symposium on}, pages 815--823. IEEE, 2010.

\bibitem{PRT12}
Mihai Patrascu, Liam Roditty, and Mikkel Thorup.
\newblock A new infinity of distance oracles for sparse graphs.
\newblock In {\em Foundations of Computer Science (focs), 2012 Ieee 53rd Annual
  Symposium on}, pages 738--747. IEEE, 2012.

\bibitem{RR98}
Ran Raz.
\newblock A parallel repetition theorem.
\newblock {\em SIAM Journal on Computing}, 27(3):763--803, 1998.

\bibitem{RS97}
Ran Raz and Shmuel Safra.
\newblock A sub-constant error-probability low-degree test, and a sub-constant
  error-probability pcp characterization of np.
\newblock In {\em Proceedings of the twenty-ninth annual ACM symposium on
  Theory of computing}, pages 475--484. ACM, 1997.

\bibitem{TZ05}
Mikkel Thorup and Uri Zwick.
\newblock Approximate distance oracles.
\newblock {\em Journal of the ACM (JACM)}, 52(1):1--24, 2005.

\bibitem{Vazirani13}
Vijay~V Vazirani.
\newblock {\em Approximation algorithms}.
\newblock Springer Science \& Business Media, 2013.

\bibitem{Wulff13}
Christian Wulff-Nilsen.
\newblock Approximate distance oracles with improved query time.
\newblock In {\em Proceedings of the Twenty-Fourth Annual ACM-SIAM Symposium on
  Discrete Algorithms}, pages 539--549. Society for Industrial and Applied
  Mathematics, 2013.

\end{thebibliography}

\clearpage

\appendix

\section{Proofs in section~\ref{sec:tz2} (lower bound of $TZ_2$-optimization problem)}\label{app:tz2}

\subsection{Label Cover Problem}

For the lower bound, we start from a hard Label Cover instance, and following the steps in proving the hardness of approximating Set Cover problem. Since the definition of the Label Cover problem is somewhat complex, we break it into parts: first defining an instance, a labelling, and then defining the problem. Note that we are using a specific setting where the parameters in the graph are strongly related, so it is slightly different from the definition of classic/general Label Cover problem.

\begin{definition}
A label cover instance consists of ($G=(V_1,V_2,E),\Sigma,\Pi$) where
\begin{itemize}
\item $G$ is a bipartite graph between vertex sets $V_1$ and $V_2$ and an edge set $E$. Let $V'=V_1\cup V_2$
\item $G$ is left and right regular. Denote by $\Delta_1$ and $\Delta_2$ the degrees of vertices in $V_1$ and $V_2$ respectively.
\item For each edge $e$, there is a constraint $\Pi_e$ which is a bijection function from $\Sigma$ to itself. The set of all constraints in $G$ are $\Pi= \{\Pi_e:\Sigma\rightarrow\Sigma\mid e\in E\}$
\end{itemize}
\end{definition}

\begin{definition}
A \emph{labelling} of the graph, is a mapping $\sigma:V'\rightarrow\Sigma$ which assigns a label for each vertex of $G$. A labelling $\sigma$ is said to satisfy an edge $e=(v_1,v_2)$ if and only if $\Pi_e(\sigma(v_1))=\sigma(v_2)$.
\end{definition}

The following definition is the problem which will be the starting point of our reduction.

\begin{definition}
In the $\text{LabelCover}_{n,r,\varepsilon}$ problem, we are given an instance $(G, \Sigma, \Pi)$ of Label Cover where $|V_1|=(5n)^r,|V_2|=(5n)^r,|\Sigma|=7^r,\Delta_1=15^r,\Delta_2=15^r$, and one of the following is true:
\begin{itemize}
\item There exists a labelling $\sigma$ such that it satisfies all the edges $e$ in $G$ (in which case we say that the input is a YES instance), or
\item For any labelling $\sigma$ of the vertices, no more than $\varepsilon^r|E|$ edges are satisfied by $\sigma$ (in which case we say that the input is a NO instance).
\end{itemize}
The goal is to determine whether the input is a YES or a NO instance.  
\end{definition}

Label Cover forms the starting point of many hardness of approximation results.  Its hardness is a now-classical application of the PCP theorem~\cite{BGKW88} and Raz's parallel repetition lemma~\cite{RR98}, which give the following theorem.

\begin{theorem}[\cite{RR98}]\label{thm:label}
There exists a constant $\varepsilon\ge0$ such that $\text{LabelCover}_{n,r,\varepsilon}$ is not in $\mathbf{P}$ unless $\mathbf{NP}\subseteq\mathbf{DTIME}(n^{O(r)})$.
\end{theorem}

For example, there exists a constant $\varepsilon\ge0$ such that $\text{LabelCover}_{n,3\log\log n,\varepsilon}$ is not in $\mathbf{P}$ unless $\mathbf{NP}\subseteq\mathbf{DTIME}(n^{O(\log\log n)})$. Starting from here, we will fix $r=3\log\log n$, and $\varepsilon$ be the constant which makes $\text{LabelCover}_{n,3\log\log n,\varepsilon}$ hard.

\subsection{$(m,l,\delta)$-Set System}

We also need a $(m,l,\delta)$-set system (see Definition~\ref{def:setsystem}) to do the reduction. We can construct a ($m,l,\delta$)-set system by using a ($l,\gamma$)-independent collection of length $m$ strings.

\begin{definition}
Let $B$ be a collection (may contains repetitions) of binary strings of length $m$. $B$ is \emph{$(l,\gamma)$-independent} if the following inequality holds for every $i_1,i_2,\mathellipsis,i_l$ and $a\in\{0,1\}^l$:
\[\left|\Pr_{x\in B}[x_{i_1}=a_1\land\mathellipsis\land x_{i_l}=a_l]-2^{-l}\right|\le\gamma.\]
\end{definition}

A corollary of lemma 1 and construction 3 in~\cite{AOJR92} provides a explict construction of ($l,\gamma$)-independent collection.

\begin{corollary}\label{cor:independent}
For any $l\le m$, there is an explict construction of a ($l,(1-2^{-l})\cdot2^{-l-1}$)-independent collection of length $m$ strings with $|B|=4^{l+1}m^2$ in $|B|^{O(1)}$ time.
\end{corollary}

With the corollary in hand, we can construct a ($m,l,\delta$)-set system with the parameters we want.

\begin{lemma}\label{lem:setsystem}
For any $l\le m$, there is an explicit construction of a ($m,l,2^{-l-1}$)-set system with $|B|=4^{l+1}m^2$ in $|B|^{O(1)}$ time.
\end{lemma}
\begin{proof}
Let $B$ be the collection of length $m$ strings in Corollary~\ref{cor:independent}. Define $C_i=\{x\in B\mid x_i=1\}$ for all $i\in[m]$, we will show that ($B;C_1,\mathellipsis,C_m$) is a ($m,l,2^{-l-1}$)-set system.

Assume that there exist $D_{i_1},D_{i_2},\mathellipsis,D_{i_l}$ such that $\left|\bigcup_{j=1}^lD_{i_j}\right|\ge(1-2^{-l-1})|B|$, where each $D_{i_j}$ is either $C_{i_j}$ or $\overline{C_{i_j}}$ (note that this implies that there are no $j$ and $k$ such that $D_{i_j}=\overline{D_{i_k}}$). Define
\[a_j=
\begin{cases} 
0,&\mbox{if }D_{i_j}=C_{i_j}\\
1,&\mbox{if }D_{i_j}=\overline{C_{i_j}}
\end{cases}.\]
Let $S=\{x\mid x\in B,x_{i_1}=a_1,x_{i_2}=a_2,\mathellipsis,x_{i_l}=a_l\}$. Because $B$ is a ($l,(1-2^{-l})\cdot2^{-l-1}$)-independent collection, we have
\[|S|=|\{x\mid x\in B,x_{i_1}=a_1,x_{i_2}=a_2,\mathellipsis,x_{i_l}=a_l\}|>(2^{-l}-(1-2^{-l})\cdot2^{-l-1})|B|>2^{-l-1}|B|.\]

On the other hand, note by construction, for all $x\in S$ and $j\in[l]$, we have $x\notin D_{i_j}$, which implies that $\left|\bigcup_{j=1}^lD_{i_j}\right|\le|B|-|S|<(1-2^{-l-1})|B|$: a contradiction.
\end{proof}

\subsubsection{Reduction}

We now show how to use the set systems from the previous section to give a reduction from Label Cover to $TZ_2$-optimization problem.

Let $(G=(V_1,V_2,E),\Sigma,\Pi)$ be a $\text{LabelCover}_{n,r,\varepsilon}$ instance with $r=3\log\log n$, and let $(B;C_1,\mathellipsis,C_m)$ be a ($m,l,2^{-l-1}$)-set system with $m=|\Sigma|=7^r,l=r\log n$.

We first create a universe $\mathcal{U}=E\times B$, and a set of sets $\mathcal{S}=\{S_{v,x}\mid v\in V',x\in[m]\}$ (recall that $V' = V_1 \cup V_2$).  Here
\[S_{v,x}=\bigcup_{e:v\in e,e\in E}\{e\}\times C_{\Pi_e(x)},\mbox{ if }v\in V_1,\]
\[S_{v,x}=\bigcup_{e:v\in e,e\in E}\{e\}\times\overline{C_x},\mbox{ if }v\in V_2.\]
We know that $|E|=(15n)^r$, $|B|=4^{r\log n+1}\cdot7^{2r}=n^{\Theta(1)\cdot r}$ and $|\mathcal{S}|=m\cdot|V'|=7^r\cdot2\cdot(5n)^r$, so $|\mathcal{U}|\gg|\mathcal{S}|$. Without lose of generality and for simplicity of our proof, we can assume $|\mathcal{U}|$ is dividable by $|\mathcal{S}|$, so that we can replicate $\mathcal{S}$ for $\frac{|\mathcal{U}|}{|\mathcal{S}|}$ times, and get a set of sets $\mathcal{S}'=\mathcal{S}^{(1)} \cup\mathellipsis\cup\mathcal{S}^{\left(\frac{|\mathcal{U}|}{|\mathcal{S}|}\right)}$ which has the same size as $\mathcal{U}$.

It is also easy to see that each $u=((v_1,v_2),b)\in\mathcal{U}$ appears in exactly $m$ sets in $\mathcal{S}$ because for each $x\in[m]$, either $u\in S_{v_1,x}$ or $u\in S_{v_2,\Pi_{(v_1,v_2)}(x)}$. Therefore each $u\in\mathcal{U}$ appears in $\frac{m|\mathcal{U}|}{|\mathcal{S}|}=\frac{|\mathcal{U}|}{|V'|}$ sets in $\mathcal{S}'$.

The metric space is defined as $V=\mathcal{U}\cup\mathcal{S}'$ and the distance is defined as following:
\[d(u,v)=
\begin{cases} 
1.2,&\mbox{if }u\in\mathcal{S}',v\in\mathcal{S}'\\
1.4,&\mbox{if }u\in v\mbox{ or }v\in u\\
1.6,&\mbox{if }u\in\mathcal{U},v\in\mathcal{U}\\
1.8,&\mbox{otherwise}
\end{cases}\]

This metric space $(V, d)$ will form the instance of $TZ_2$-optimization which we analyze. It is easy to see that the reduction is polynomial because $|V|$ is polynomial of $|E|$.

\subsubsection{Analysis}

\begin{lemma}\label{lem:tz2yes}
If ($G,\Sigma,\Pi$) is a YES instance in the $\text{LabelCover}_{n,r,\varepsilon}$ problem. Then the reduction $(V,d)$ to the $TZ_2$-optimization problem has a solution with cost $\le(|V'|+1)\cdot|V|$.
\end{lemma}

\begin{proof}
Let $\sigma:V'\rightarrow[m]$ denote a labelling of G which satisfies all the edges in $E$. Let $A_1=\{S_{v,\sigma(v)}\mid v\in V'\}$. We claim that $A_1$ is a solution with cost at most $(|V'|+1)\cdot|V|$.

Note that in Section~\ref{sec:utz2} we showed that the level $2$ cost $\sum_{u\in V}|R_{2u}|=|V|\cdot|A_1|=|V'|\cdot|V|$, the only thing left is to show that the level $1$ cost is $\sum_{u\in V}|R_{1u}|\le|V|$. We will prove this by showing $R_{1u}\subseteq\{u\}$ for all $u\in V$.

For any $u\in\mathcal{S}'$, we have that $d(u,v)\ge1.2$ for all $v\in V$ because the definition of $d$, and $\min_{w\in A_1}d(u,w)=1.2$ because $A_1\cap\mathcal{S}'\ne\varnothing$. Thus $R_{1u}=\{v\in V\mid d(u,v)<\min_{w\in A_1}d(u,w)\}\subseteq\{u\}$.

For any $u=((v_1,v_2),b)\in\mathcal{U}$, we have that $d(u,v)\ge1.4$ for all $v\in V$ because the definition of $d$. We also know that either $u\in S_{v_1,\sigma(v_1)}$ or $u\in S_{v_1,\Pi_{(v_1,v_2)}(\sigma(v_1))}$ by the definition of $\mathcal{S}$, and $\Pi_{(v_1,v_2)}(\sigma(v_1))=\sigma(v_2)$ because edge $(v_1,v_2)$ is satisfied by labelling $\sigma$. Therefore $u\in S_{v_1,\sigma(v_1)}\cup S_{v_2,\sigma(v_2)}$. From the fact that both $S_{v_1,\sigma(v_1)}$ and $S_{v_2,\sigma(v_2)}$ are in $A_1$, we have $\min_{w\in A_1}d(u,w)=1.4$. Thus $R_{1u}=\{v\in V\mid d(u,v)<\min_{w\in A_1}d(u,w)\}\subseteq\{u\}$.

Therefore $R_{1u}\subseteq\{u\}$ for all $u\in V$, so that $A_1$ is a solution with cost at most $\le(|V'|+1)\cdot|V|$.
\end{proof}

\begin{lemma}\label{lem:tz2no}
If ($G,\Sigma,\Pi$) is a No instance in the $\text{LabelCover}_{n,r,\varepsilon}$ problem. Then the reduction $(V,d)$ to the $TZ_2$-optimization problem has no solution with cost $<\frac{l}{8}|V'|\cdot|V|$.
\end{lemma}

\begin{proof}
We prove the lemma by showing that if the optimal solution of the reduction $(V,d)$ to the $TZ_2$-optimization problem has cost $<\frac{l}{8}|V'|\cdot|V|$, then there exists a labelling $\sigma$ such that it satisfies more than $\varepsilon^r|E|$ edges.

Assume the optimal solution $A_1\subseteq V$ has $cost(A_1,V,d)<\frac{l}{8}|V'|\cdot|V|$, then $|A_1|<\frac{l}{8}|V'|$ because the level $2$ cost is $\sum_{u\in V}|R_{2u}|=|V||A_1|$.

Let $L_v=\{x\in[m]\mid\exists j\in\left[\frac{|\mathcal{U}|}{|\mathcal{S}|}\right]\text{ s.t. }S_{v,x}^{(j)}\in A_1\cap\mathcal{S}'\}$ for all $v\in V$, then $\sum_vL_v\le|A_1\cap\mathcal{S}'|\le|A_1|<\frac{l}{8}|V'|$. Therefore at least $\frac{3}{4}|V'|$ vertices has $|L_v|<\frac{l}{2}$, because otherwise $\sum_vL_v\ge(1-\frac{3}{4})\cdot|V'|\cdot\frac{l}{2}\ge\frac{l}{8}|V'|$.

Let $E_1=\{e=(v_1,v_2)\in E\mid|L_{v_1}|<\frac{l}{2},|L_{v_2}|<\frac{l}{2}\}$. Then $|E_1|\ge\frac{|E|}{2}$ because $|V_1|=|V_2|=\frac{|V'|}{2}$ and $G$ is regular.

On the other hand, we define a $u\in\mathcal{U}$ is ``uncovered'' if $\{v\in A_1\cap\mathcal{S}'\mid u\in v\}=\varnothing$. Then for any uncovered $u\in\mathcal{U}$, we know that $\min_{w\in A_1}d(u,w)=1.6$. Thus
\begin{align*}
R_{1u}=&\{v\in V\mid d(u,v)<\min_{w\in A_1}d(u,w)\}\\
\ge&\{v\in\mathcal{S}'\mid d(u,v)<1.6\}\\
=&\{v\in\mathcal{S}'\mid u\in v\}
=\frac{|\mathcal{U}|}{|V'|}.
\end{align*}

Therefore $|\{u\in\mathcal{U}\mid u\text{ is uncovered}\}|<\frac{\frac{l}{8}|V'|\cdot|V|}{\frac{|\mathcal{U}|}{|V'|}}<\frac{l}{4}|V'|^2$.

Let $E_2=\{e\in E\mid|\{u=(e,b)\in\mathcal{U}\mid u\text{ is uncovered}\}|<\frac{l|V'|^2}{|E|}\}$. Then $|E_2|\ge\frac{3}{4}|E|$ because otherwise $|\{u\in\mathcal{U}\mid u\text{ is uncovered}\}|\ge(|E|-\frac{3}{4}|E|)\cdot\frac{l|V'|^2}{|E|}\}\ge\frac{l}{4}|V'|^2$.

Let $E'=E_1\cap E_2$, we know that $|E'|\ge\frac{|E|}{4}$.

Now, we will show that if we uniformly random sample labels from $L_v$ for each $v\in V'$, the expected number of the edges satisfied in $E'$ is at least $\frac{|E|}{l^2}$.

For each edge $e=(v_1,v_2)\in E'$ where $v_1\in V_1$ and $v_2\in V_2$. Assume $L_{v_1}=\{a_1,\mathellipsis,a_p\}$, $L_{v_2}=\{b_1,\mathellipsis,b_q\}$. Note that for every $e\in E_2$ we have
\[\Big|\{u=(e,b)\in\mathcal{U}\mid\exists v\in A_1\cap\mathcal{S}',u\in v\}\Big|\ge|B|-\frac{l|V'|^2}{|E|},\]
and for all $u=((v_1,v_2),b)\in\mathcal{U}$, there exists $v\in A_1\cap\mathcal{S}'$ where $u\in v$ iff $u\in S_{v_1,a_i}$ or $u\in S_{v_2,b_i}$. Thus we have
\[\left|(\{e\}\times B)\cap\left((\bigcup_{i=1}^pS_{v_1,a_i})\cup(\bigcup_{j=1}^qS_{v_2,b_j})\right)\right|\ge|B|-\frac{l|V'|^2}{|E|},\]
which means
\[\left|(\bigcup_{i=1}^pC_{\Pi_e(a_i)})\cup(\bigcup_{j=1}^q\overline{C_{b_j}})\right|\ge|B|-\frac{l|V'|^2}{|E|}=(1-\frac{l|V'|^2}{|E||B|})|B|=(1-(\frac{5}{147n})^rl)|B|>(1-2^{-l-1})|B|.\]
Thus by the definition of ($m,l,2^{-l-1})$)-set system, we know that there exists $i,j$ such that $\Pi_e(a_i)=b_j$. Therefore, $e$ is satisfied with probability $\frac{1}{|L_{v_1}|\cdot|L_{v_2}|}\ge\frac{4}{l^2}$ because the labels are uniformly sampled. Thus the expected number of the edges satisfied in $E'$ is at least $\frac{4}{l^2}\cdot\frac{|E|}{4}=\frac{|E|}{l^2}$, which means, there is a way to label all the vertices in $V'$ and satisfies at least $\frac{|E|}{l^2}$ edges.

Finally, because $r=3\log\log n$ and $l=r\log n$, we know that at most $\varepsilon^r\cdot|E|<\frac{|E|}{l^2}$ edges can be satisfied by any labelling, which is a contradiction.
\end{proof}

With these lemmas, we can prove our lower bound on the $TZ_2$-optimization problem.\\
\\
\textbf{Proof of Theorem}~\ref{thm:ltz2}:
By Lemma~\ref{lem:tz2yes} and Lemma~\ref{lem:tz2no}, we have a polynomial reduction from $\text{LabelCover}_{n,r,\varepsilon}$ problem to $TZ_2$-optimization problem, which maps a YES instance of $\text{LabelCover}_{n,r,\varepsilon}$ to a $TZ_2$-optimization instance with optimal cost at most $(|V'|+1)\cdot|V|$, and maps a NO instance of $\text{LabelCover}_{n,r,\varepsilon}$ to a $TZ_2$-optimization instance with optimal cost at least $\frac{l}{8}|V'|\cdot|V|$. The gap is $\frac{\frac{l}{8}|V'|\cdot|V|}{(|V'|+1)\cdot|V|}=\Theta(\frac{l}{8})=\Theta(\log|V|)$.

Combined with the hardness Theorem~\ref{thm:label}, we know that unless $\mathbf{NP}\subseteq\mathbf{DTIME}(n^{O(\log\log n)})$, the $TZ_2$-optimization problem does not admit a polynomial-time $o(\log n)$-approximation.
\qed

\section{Proofs in section~\ref{sec:tzk} (integrality gap of $TZ_k$-optimization problem)}\label{app:tzk}

\subsection{Relaxation validity}

We first prove that our LP relaxation is indeed valid, i.e., we prove the following claim.

\begin{claim}\label{claim:lptzk}
$LP_{TZ_k}$ is a valid relaxation to the $TZ_k$-optimization problem.
\end{claim}
\begin{proof}
Let $A_1,\mathellipsis,A_{k-1}$ be a valid solution to the $TZ_k$-optimization problem. Let $x_v^{(i)}=\mathds{1}_{v\in A_i}$ and $y_{uv}^{(i)}=\mathds{1}_{v\in R_{iu}}$ for all $i\in[k]$ and $u,v\in V$. We can see that the objective value $\sum_{i=1}^k\sum_{u,v\in V}y_{uv}^{(i)}=\sum_{i=1}^k\sum_{u\in V}|R_{iu}|=cost(A_1,\mathellipsis,A_{k-1},V,d)$, which is the cost function. 

We can also see that the first constraint is satisfied by $x_v^{(i)}$ and $y_{uv}^{(i)}$ because $\varnothing=A_k\subseteq A_{k-1}\subseteq\mathellipsis\subseteq A_0=V$. The second constraint is satisfied because if $v\in A_{i-1}$ and there is no vertex in $A_i\cap B_u(v)$, then $v\in R_{iu}$. The third constraint is trivially satisfied.

Therefore $x_v^{(i)}$ and $y_{uv}^{(i)}$ is a valid solution to $LP_{TZ_k}$ which makes the LP objective value equal to the actuall cost function. Thus the claim is proved.
\end{proof}

\subsection{Integrality gap}

Let's consider an instance $(V,d)$ with $V=[n]$. All $n$ vertices lie on a circle and they evenly split the cycle. The cycle distance $d(u,v)=\min\{|u-v|,n+\min\{u,v\}-\max\{u,v\}\}$.

We first show that on this instance, $LP_{TZ_k}$ has a solution with low cost.

\begin{lemma}\label{lem:gapu}
$LP_{TZ_k}$ has a solution with cost $O(n^{1+\frac{1}{2^{k-1}}})$ on instance $(V,d)$.
\end{lemma}
\begin{proof}
Consider the following setting of the LP variables: let $x_{v}^{(i)}=n^{-\frac{2^i-1}{2^{k-1}}}$ for all $v\in V$ and $i\in[k-1]$, and let $y_{uv}^{(i)}=\max\{0,x_v^{(i-1)}-\sum_{w\in B_u(v)}x_w^{(i)}\}$ for all $u,v\in V$ and $i\in[k]$.

We can see that $x_{v}^{(i)}=n^{-\frac{2^i-1}{2^{k-1}}}\ge n^{-\frac{2^{i+1}-1}{2^{k-1}}}=x_{v}^{(i+1)}$ which satisfies the first constraint of $LP_{TZ_k}$, $y_{uv}^{(i)}\ge x_v^{(i-1)}-\sum_{w\in B_u(v)}x_w^{(i)}$ which satisfies the second constraint of $LP_{TZ_k}$, and $y_{uv}^{(i)}\ge0$ which satisfies the third constraint of $LP_{TZ_k}$. Therefore $x_{v}^{(i)},y_{uv}^{(i)}$ is a valid solution to $LP_{TZ_k}$.

The objective value of this solution is
\begin{align}
\sum_{i=1}^k\sum_{u,v\in V}y_{uv}^{(i)}=&\sum_{i=1}^k\sum_{u,v\in V}\max\{0,x_v^{(i-1)}-\sum_{w\in B_u(v)}x_w^{(i)}\}\\
=&\sum_{i=1}^{k-1}\sum_{u,v\in V}\max\{0,x_v^{(i-1)}-|B_u(v)|x_v^{(i)}\}+\sum_{u,v\in V}(x_v^{(k-1)}-0)\label{eqn:gap1}\\
=&\sum_{i=1}^{k-1}\sum_{u,v\in V}\max\{0,n^{-\frac{2^{i-1}-1}{2^{k-1}}}-(2\cdot d(u,v)+1)\cdot n^{-\frac{2^i-1}{2^{k-1}}}\} +\sum_{u,v\in V}n^{-\frac{2^{k-1}-1}{2^{k-1}}}\label{eqn:gap2}\\
=&\sum_{i=1}^{k-1}\sum_{u\in V}\Big(n^{-\frac{2^{i-1}-1}{2^{k-1}}}-n^{-\frac{2^i-1}{2^{k-1}}}+n^{-\frac{2^{i-1}-1}{2^{k-1}}}-3\cdot n^{-\frac{2^i-1}{2^{k-1}}}+\mathellipsis\Big) +n^{1+\frac{1}{2^{k-1}}}\label{eqn:gap3}\\
=&\sum_{i=1}^{k-1}n\cdot\Big(O(n^{-\frac{2^{i-1}-1}{2^{k-1}}})\cdot O(n^{\frac{2^i-1}{2^{k-1}}-\frac{2^{i-1}-1}{2^{k-1}}})\Big)+n^{1+\frac{1}{2^{k-1}}}\label{eqn:gap4}\\
=&\sum_{i=1}^{k-1}O(n^{1+\frac{1}{2^{k-1}}})+n^{1+\frac{1}{2^{k-1}}}\\
=&O(n^{1+\frac{1}{2^{k-1}}})
\end{align}

Here equation~(\ref{eqn:gap1}) holds because of all $x_v^{(i)}$ are equal and all $x_v^{(k)}=0$. Equation~(\ref{eqn:gap2}) holds because of the definition of circle distance. Equation~(\ref{eqn:gap3}) is a unrolling, and equation~(\ref{eqn:gap4}) is a summation over arithmetic progression. The last equation holds because of $k$ is a constant.
\end{proof}

Next we will show that the optimal solution of this instance is large.

\begin{lemma}\label{lem:gapl}
The optimal solution to the instance $(V,d)$ has cost at least $\Omega(n^{1+\frac{1}{k}})$.
\end{lemma}

We will prove this lemma using a stronger claim. The lemma holds by setting $a=1,b=\lfloor\frac{n}{2}\rfloor$, and $l=k$ in this claim:

\begin{claim}
For a segment $[a,b]$ of the cycle where $a,b\in[n]$, $b-a<\frac{n}{2}$, and all the vertices in $[a,b]$ are NOT in $A_l$, we have $\sum_{i=1}^l\sum_{u\in[a,b]\cap[n]}|R_{iu}|\ge\left(\frac{b-a+1}{4^l}\right)^{1+\frac{1}{l}}$ for each $l\in[k]$.
\end{claim}
\begin{proof}
We prove this by doing induction on $l$. The base case is $l=1$. For each vertex $u\in[a,b]$, We know that
\[R_{1u}=\{v\in V\mid d(u,v)<\min_{w\in A_1}d(u,w)\}\subseteq\{v\in[a,b]\mid|u-v|\le\min\{u-a,b-u\}\},\]
so $|R_{1u}|\ge2\cdot\min\{u-a,b-u\}$ because all the vertices in $[a,b]$ are NOT in $A_1$. So
\[\sum_{u\in[a,b]}|R_{1u}|\ge2\cdot(1+2+\mathellipsis+\left\lfloor\frac{b-a+2}{2}\right\rfloor+\mathellipsis+2+1)\ge\left(\frac{b-a+1}{4}\right)^2.\]
Now we consider general case $l\ge2$, and assume the claim is established on $l-1$.

Assume there are $m$ vertices $t_1,\mathellipsis,t_m\in[a,\lceil\frac{a+b}{2}\rceil]\cap A_{l-1}$, and $[a,\lceil\frac{a+b}{2}\rceil]$ are splitted to small segments $[a_0,b_0],\mathellipsis,[a_m,b_m]$ where all the vertices in $[a_i,b_i]$ are not in $A_{l-1}$ (if a segment has no vertex inside, we let $b_i=a_i-1$ without lose of generality). Then for each $i\in[m]$ and $u\in[a_i,b_i+1]$, we have $t_1,\mathellipsis,t_i\in R_{lu}$ because $R_{lu}=\{v\in A_{l-1}\mid d(u,v)<\min_{w\in A_l}d(u,w)\}$. Thus
\begin{align*}
\sum_{i=1}^l\sum_{u\in[a,b]\cap[n]}|R_{iu}|\ge&\sum_{i=0}^m\left(\sum_{j=1}^{l-1}\sum_{u\in[a_i,b_i]\cap[n]}|R_{ju}|+\sum_{u\in[a_i,b_i]\cap[n]}|R_{lu}|\right)\\
\ge&\sum_{i=0}^m\left(\left(\frac{b_i-a_i+1}{4^{l-1}}\right)^{1+\frac{1}{l-1}}+i\cdot(b_i-a_i+2)\right).
\end{align*}

If $m>\frac{b-a+1}{4}$, $\sum_{i=0}^mi\cdot1$ is already at least $\left(\frac{b-a+1}{4^l}\right)^{1+\frac{1}{l}}$.

If $m\le\frac{b-a+1}{4}$, we have a stronger inequality which we will prove later:

\begin{lemma}\label{lem:ineq}
If $\alpha\in[1,2]$ and $x_i\ge0$ for all $i\in[m]$, then
\[\sum_{i=0}^m(x_i^\alpha+4i\cdot x_i)\ge\left(\sum_{i=0}^mx_i\right)^{2-\frac{1}{\alpha}}\]
\end{lemma}

Using this inequality, by setting $x_i=\frac{b_i-a_i+1}{4^{l-1}}$ and $\alpha=1+\frac{1}{l-1}$ we have
\begin{align*}
\sum_{i=1}^l\sum_{u\in[a,b]\cap[n]}|R_{iu}|\ge&\left(\sum_{i=0}^m\frac{b_i-a_i+1}{4^{l-1}}\right)^{2-\frac{1}{1+\frac{1}{l-1}}}\\
\ge&\left(\frac{\lceil\frac{a+b}{2}\rceil-a+1-m}{4^{l-1}}\right)^{1+\frac{1}{l}}\\
\ge&\left(\frac{\frac{b-a}{2}+1-\frac{b-a}{4}}{4^{l-1}}\right)^{1+\frac{1}{l}}\ge\left(\frac{b-a+1}{4^l}\right)^{1+\frac{1}{l}}.
\end{align*}
\end{proof}

With these lemma in hand, we can now prove Theorem~\ref{thm:gaptzk}.
\\
\text{}\\
\textsl{Proof of Theorem~\ref{thm:gaptzk}:}

Combine Lemma~\ref{lem:gapu} and Lemma~\ref{lem:gapl}, there is an $\Omega(\frac{n^{1+\frac{1}{k}}}{n^{1+\frac{1}{2^k-1}}})=\Omega(n^{\frac{1}{k}-\frac{1}{2^k-1}})$ integrality gap for the basic LP relaxation $LP_{TZ_k}$.
\qed
\\
\text{}\\
\textsl{Proof of Lemma~\ref{lem:ineq}:}

Let $M=(\sum_{i=1}^mx_i)^\frac{\alpha-1}{\alpha}$. We first split the problem to 2 cases, depending on whether $m\le M$.

\textbf{Case 1:} $m\le M$.

In this case, by H\"older's inequality, we have
\[(\sum_{i=0}^mx_i^\alpha)^\frac{1}{\alpha}\cdot(\sum_{i=0}^m1^\frac{\alpha}{\alpha-1})^\frac{\alpha-1}{\alpha}\ge(\sum_{i=0}^mx_i\cdot 1)\]
thus
\[\sum_{i=0}^mx_i^\alpha\ge\frac{(\sum_{i=0}^mx_i)^\alpha}{m^{\alpha-1}}\ge\frac{(\sum_{i=0}^mx_i)^\alpha}{M^{\alpha-1}}\ge(\sum_{i=0}^mx_i)^{\alpha-\frac{\alpha-1}{\alpha}\cdot(\alpha-1)}=(\sum_{i=0}^mx_i)^{2-\frac{1}{\alpha}}\]

\textbf{Case 2:} $m>M$.

Let's fix $T=\sum_{i=0}^mx_i$ and consider the $\mathbf{x}^*$ which minimizes the left side: $l(\mathbf{x})=\sum_{i=0}^m(x_i^\alpha+4i\cdot x_i)$.

Consider any two consecutive variables $x_j$ and $x_{j+1}$, we claim that, in $\mathbf{x}^*$, for each $0\le j<m$, either $x_{j+1}^*=0$, or $(x_j^*)^{\alpha-1}-(x_{j+1}^*)^{\alpha-1}=\frac{4}{\alpha}$.

This is because, if we replace the $x_{j+1}$ in $l(\mathbf{x})$ by $T-\sum_{i\ne(j+1)}x_i$ and do partial derivative with respect of $x_j$, we have
\begin{align*}
&\frac{\partial}{\partial x_j}\left(\sum_{i\ne(j+1)}(x_i^\alpha+4i\cdot x_i)+\left(T-\sum_{i\ne(j+1)}x_i\right)^\alpha+4(j+1)\cdot\left(T-\sum_{i\ne(j+1)}x_i \right)\right)\\
=&\frac{\partial}{\partial x_j}\left(x_j^\alpha+4j\times x_j+\left(T-\sum\nolimits_{i\ne(j+1)}x_i\right)^\alpha+4(j+1)\cdot\left(T-\sum\nolimits_{i\ne(j+1)}x_i \right)\right)\\
=&\alpha\cdot x_j^{\alpha-1}+4j-\alpha\cdot\left(T-\sum\nolimits_{i\ne(j+1)}x_i \right)^{\alpha-1}-4(j+1)\\
=&\alpha\cdot\left(x_j^{\alpha-1}-\left(T-\sum\nolimits_{i\ne(j+1)}x_i \right)^{\alpha-1}\right)-4\\
=&\alpha(x_j^{\alpha-1}-x_{j+1}^{\alpha-1})-4.
\end{align*}
If we fix $x_i$ for all $i\in[0,m]\cap\mathbb{N}\backslash\{j,j+1\}$, this partial derivative monotonically increases as $x_j$ increases. Thus when $l(\mathbf{x})$ is minimized, either the partial derivative equals $0$, which means $(x_j^*)^{\alpha-1}-(x_{j+1}^*)^{\alpha-1}=\frac{4}{\alpha}$, or $x_j$ hits the ceiling, which means $x_j^*=T-\sum_{i\ne j,(j+1)}x_i^*$, so $x_{j+1}^*=0$.

This result shows that, the number series $(x_0^*)^{\alpha-1},(x_1^*)^{\alpha-1},\mathellipsis,(x_m^*)^{\alpha-1}$ is in decreasing order, where
\[(x_i^*)^{\alpha-1}=
\begin{cases}
(x_{i-1}^*)^{\alpha-1}-\frac{4}{\alpha},&\mbox{if }(x_i^*)^{\alpha-1}>\frac{4}{\alpha}\\
0,&\mbox{otherwise}
\end{cases}\]

If the number of non-zero entries in $\mathbf{x}^*$ is at most $M$, then this comes back to the Case 1. Otherwise, there are more than $M$ non-zero entries in $\mathbf{x}^*$, thus $x_0,x_1,\mathellipsis,x_M$ are all non-zero, and $(x_i^*)^{\alpha-1}\ge\frac{4}{\alpha}\cdot(M-i)$ for $i\le M$. Therefore
\begin{align}
\sum_{i=0}^m(x_i^*)^\alpha&\ge\sum_{i=0}^M\left(\frac{4}{\alpha}\cdot(M-i)\right)^\frac{\alpha}{\alpha-1}\\
&\ge\sum_{i=1}^M\left(\frac{4i}{\alpha}\right)^\frac{\alpha}{\alpha-1}\\
&\ge\frac{(\sum_{i=1}^M\frac{4i}{\alpha})^\frac{\alpha}{\alpha-1}}{M^\frac{1}{\alpha-1}}\label{eqn:holder}\\
&\ge\frac{(\frac{2M^2}{\alpha})^\frac{\alpha}{\alpha-1}}{M^\frac{1}{\alpha-1}}\\
&\ge(\frac{2}{\alpha})^\frac{\alpha}{\alpha-1}\cdot(\sum_{i=0}^mx_i)^{\frac{\alpha-1}{\alpha}\cdot(\frac{2\alpha}{\alpha-1}-\frac{1}{\alpha-1})}\\
&\ge(\sum_{i=0}^mx_i)^{2-\frac{1}{\alpha}}\label{eqn:alpha}
\end{align}
Here inequality~(\ref{eqn:alpha}) holds because of $\alpha\in[1,2]$. Inequality~(\ref{eqn:holder}) holds because of H\"older's inequality
\[(\sum_{i=1}^m1^\alpha)^\frac{1}{\alpha}\cdot(\sum_{i=1}^my_i^\frac{\alpha}{\alpha-1})^\frac{\alpha-1}{\alpha}\ge\sum_{i=1}^my_i\cdot 1\]
\qed

\section{Proofs in section~\ref{sec:pr}}\label{app:pr}

\subsection{Proof of Valid Relaxation}

We prove the following claim:

\begin{claim}\label{claim:lppr}
$LP_{PR}$ is a valid relaxation to the $PR$-optimization problem.
\end{claim}

Let $A$ be a valid solution to the $PR$-optimization problem. Let $x_v=\mathds{1}_{v\in A}$ and $y_{uv}=\mathds{1}_{\{u,v\}\in R}$ for all $u,v\in V$. We can see that the objective value $\sum_{v\in V}n\cdot x_v+\sum_{\{u,v\}\subseteq V}y_{uv}=n\cdot|A|+|R|=cost(A,V,d)$, which is the cost function. 

We can also see that the first constraint is satisfied by $x_v$ and $y_{uv}$ because if $y_{uv}=0$, we have $d(u,v)\ge\min_{w\in A}d(u,w)+\min_{w\in A}d(v,w)-1$, then for all $r\in[0,d(u,v)]$, there must be a vertex in $A\cap(B_u(r)\cup B_v(d(u,v)-r))$, which makes $0\ge1-\sum_{w\in B_u(r)\cup B_v(d(u,v)-r)}x_w$ satisfied. The second and the third constraints are trivially satisfied.

Therefore $x_v$ and $y_{uv}$ is a valid solution to $LP_{PR}$ which makes the LP objective value equal to the actuall cost function. Thus the claim is proved.
\qed

\subsection{Lower Bound Proofs}

We start from the following theorem:

\begin{theorem}[\cite{RS97}]\label{thm:setcover}
Unless $\mathbf{P}=\mathbf{NP}$, there is no $o(\log n)$-approximation to the set cover problem. 
\end{theorem}

We now prove two lemmas about the reduction (completeness and soundness).

\begin{lemma}
If there is a solution $\mathcal{S}^*$ to the set cover instance $(\mathcal{U},\mathcal{S})$ where $|\mathcal{S}^*|=t$, then there is a set $A$ where $cost(A,V,d)\le t|V|$.
\end{lemma}
\begin{proof}
For each $S\in S^*$, we add an arbitrary element from $G_S$ to $A$. Then for every vertex in $V$, the closest vertex in $A$ has distance at most $1$ to it. Therefore
\begin{align*}
R=&\left\{\{u,v\}\subseteq V\mid d(u,v)<\min_{w\in A}d(u,w)+\min_{w\in A}d(v,w)-1\right\}\\
=&\left\{\{u,v\}\subseteq V\mid d(u,v)<1+1-1\right\}=\varnothing
\end{align*}
Thus the total cost is at most $|V|\cdot|A|+|R|=t|V|$.
\end{proof}
\begin{lemma}
If there is a set $A \subseteq V$ where $cost(A,V,d)\le t|V|$, then there exists a solution $\mathcal{S^*}$ to the set cover instance $(\mathcal{U},\mathcal{S})$ where $|\mathcal{S}^*|=t$.
\end{lemma}
\begin{proof}
First, we say that a group $G=G_e$ or $G=G_S$ is ``covered'' if there exists a vertex $u\in G$, which $\min_{w\in A}d(u,w)=1$. Then by the definition of $d$, it's easy to see that if a group $G$ is covered, then for all vertices $u\in G$, we have $\min_{w\in A}d(u,w)=1$. In addition, if a group $G_e$ is covered, then either $G_e\cap A\ne\varnothing$, or there is a $S\in\mathcal{S}$, where $e\in S$ and $G_S\cap A\ne\varnothing$.

We can also see that, if a group $G$ is not covered, then let
\begin{align*}
R_G=&\left\{\{u,v\}\subseteq G\mid d(u,v)<\min_{w\in A}d(u,w)+\min_{w\in A}d(v,w)-1\right\}\\
=&\left\{\{u,v\}\subseteq G\mid d(u,v)<2+2-1\right\}\\
=&\{\{u,v\}\subseteq G\}
\end{align*}
Thus $|R_G|\ge\frac{3n(3n-1)}{2}>3n^2>|V|$. Therefore if we add an arbitrary element from $G$ to $A$, then $|R|$ decreases by at least $|R_G|\ge|V|$, and $|V|\cdot|A|$ increases by $|V|$, which makes $cost(A,V,d)$ only decrease. If we keep doing this operation, there will be a set $A$ which makes sure that all the groups are covered, and $cost(A,V,d)\le t|V|$. Now, for every vertex $u\in V$, we have $\min_{w\in A}d(u,w)=1$, so $R=\varnothing$.

We can keep modifying $A$ to the form we want. If there is a vertex $v\in G_e\cap A$, removing $v$ and simultaneously adding a vertex in any $S\ni e$ to $A$ does not increase the cost. This is because this operation keeps the fact that all the groups are covered.

Finally, we have a set $A$ where only contains vertices in $\bigcup_{S\in\mathcal{S}}G_S$ and $cost(A,V,d)\le t|V|$. Let $\mathcal{S}^*=\{S\in\mathcal{S}\mid G_S\cap A\ne\varnothing\}$.  Then $|\mathcal{S}^*|\le t$ because $cost(A,V,d)=|A|\cdot|V|+|R|\ge|\mathcal{S}^*|\cdot|V|$, and $\mathcal{S}^*$ covers $\mathcal{U}$ because all the group $G_e$ are covered.
\end{proof}

These lemmas, combined with Theorem~\ref{thm:setcover}, imply Theorem~\ref{thm:lpr}

\section{Proofs in section~\ref{sec:doo}}\label{app:doo}

\subsection{$TZ_2$-Optimization Problem With Outliers}

\subsubsection{Proof of expected cost}

\begin{lemma}
If $\|\vec{y}_{uv}^*\|^2\le\frac{\varepsilon}{2}$, then the probability that $uv\in R_{1u}$ is at most $\frac{1}{n}$.
\end{lemma}
\begin{proof}
If $\|\vec{z}_u^*\|^2\ge\frac{1}{1+\varepsilon}$ or $\|\vec{z}_v^*\|^2\ge\frac{1}{1+\varepsilon}$, then $u$ or $v$ is in $F$, so $v\notin R_{1u}$. Thus we only consider the case that $\|\vec{z}_u^*\|^2<\frac{1}{1+\varepsilon}$ and $\|\vec{z}_v^*\|^2<\frac{1}{1+\varepsilon}$, which means $\vec{z}_u^*\cdot\vec{z}_v^*<\frac{1}{1+\varepsilon}$. Since $\|\vec{y}_{uv}^*\|^2\le\frac{\varepsilon}{2}$, we have
\[\sum_{w\in B_u(v)}\|\vec{x}_w^*\|^2\ge1-\frac{\varepsilon}{2}-\frac{1}{1+\varepsilon}\ge\frac{\varepsilon}{3}.\]

Therefore, the probability that $d(u,v)<\min_{w\in A}d(u,w)$ is at most
\[\prod_{w\in B_{u}(v)}(1-\min\{\frac{3\ln n}{\varepsilon}\cdot\|\vec{x}_v^*\|^2,1\})\le e^{-\sum_{w\in B_{u}(v)}\frac{3\ln n}{\varepsilon}\cdot\|\vec{x}_v^*\|^2}\le\frac{1}{n}\]
\end{proof}

Therefore, let $OPT_{SDP_{TZ_2O}}$ denotes the optimal cost of $SDP_{TZ_2O}$, then the expected cost of the rounding algorithm is at most
\[\sum_{v\in V}(n-f)\cdot\frac{3\ln n}{\varepsilon}\cdot\|\vec{x}_v^*\|^2+\frac{2}{\varepsilon}\cdot\sum_{u,v\in V}\|\vec{y}_{uv}^*\|^2+n^2\cdot\frac{1}{n}\le O(\log n)\cdot SDP_{TZ_2O}+n\le O(\log n)\cdot OPT\]
because $OPT\ge\Omega(n)$, which proves Theorem~\ref{thm:tz2o}.

\subsubsection{True approximation}

When the number of outliers is low, in particular when $f \leq \sqrt{n}$, we can find an actual $O(\log n)$-approximation.

The SDP and rounding algorithm are the same, except we will choose $f$ vertices with the highest $\|\vec{z}_v^*\|^2$ values as $F$, rather than a threshold rounding of $\frac{1}{1+\varepsilon}$.

Now there are two cases when $\|\vec{y}_{uv}^*\|^2\le\frac{\varepsilon}{2}$. One case is the same as before, where $\sum_{w\in B_u(v)}\|\vec{x}_w^*\|^2\ge\frac{\varepsilon}{3}$. In this case, the probability that $v\in R_{1u}$ is at most $\frac{1}{n}$. The other case is $\sum_{w\in B_u(v)}\|\vec{x}_w^*\|^2<\frac{\varepsilon}{3}$, which means $\vec{z}_u^*\cdot\vec{z}_v^*\ge1-\frac{\varepsilon}{2}-\frac{\varepsilon}{3}=1-\frac{5}{6}\varepsilon$.

However, this case will not appear a lot. Whenever $\vec{z}_u^*\cdot\vec{z}_v^*\ge1-\frac{5}{6}\varepsilon$, both $\|\vec{z}_u^*\|$ and $\|\vec{z}_v^*\|$ should be at least $1-\frac{5}{6}\varepsilon$, which means $\|\vec{z}_u^*\|^2$ and $\|\vec{z}_v^*\|^2$ is at least $\frac{1}{2}$. Because $\sum_{v\in V}\|\vec{z}_v^*\|^2\le f$, we know that there are at most $2f$ of $\|\vec{z}_v^*\|^2$ are at least $\frac{1}{2}$. Therefore the number of $u,v$ pairs that $\|\vec{y}_{uv}^*\|^2\le\frac{\varepsilon}{2}$ and $\sum_{w\in B_u(v)}\|\vec{x}_w^*\|^2<\frac{\varepsilon}{3}$ is at most $2f\cdot2f=4n$.

Therefore, let $OPT_{SDP_{TZ_2O}}$ denotes the optimal cost of $SDP_{TZ_2O}$, then the expected cost of the rounding algorithm is at most
\[\sum_{v\in V}(n-f)\cdot\frac{3\ln n}{\varepsilon}\cdot\|\vec{x}_v^*\|^2+\frac{2}{\varepsilon}\cdot\sum_{u,v\in V}\|\vec{y}_{uv}^*\|^2+n^2\cdot\frac{1}{n}+4n\le O(\log n)\cdot OPT_{SDP_{TZ_2O}}+5n\le O(\log n)\cdot OPT\]
because $OPT\ge\Omega(n)$, which proves Theorem~\ref{thm:tz2ot}.

\subsection{$PR$-Optimization Problem With Outliers}

\begin{lemma}
If $\|\vec{y}_{uv}^*\|^2\le\frac{\varepsilon}{2}$, then the probability that $\{u,v\}\in R$ is at most $\frac{1}{n}$.
\end{lemma}
\begin{proof}
If $\|\vec{z}_u^*\|^2\ge\frac{1}{1+\varepsilon}$ or $\|\vec{z}_v^*\|^2\ge\frac{1}{1+\varepsilon}$, then $u$ or $v$ is in $F$, so $\{u,v\}\notin R$. Thus we only consider the case that $\|\vec{z}_u^*\|^2<\frac{1}{1+\varepsilon}$ and $\|\vec{z}_v^*\|^2<\frac{1}{1+\varepsilon}$, which means $\vec{z}_u^*\cdot\vec{z}_v^*<\frac{1}{1+\varepsilon}$. Since $\|\vec{y}_{uv}^*\|^2\le\frac{\varepsilon}{2}$, we have
\[\sum_{w\in B_u(r)\cup B_v(d(u,v)-r)}\|\vec{x}_w^*\|^2\ge1-\frac{\varepsilon}{2}-\frac{1}{1+\varepsilon}\ge\frac{\varepsilon}{3}.\]

Therefore, the probability that $A\cap(B_u(r)\cup B_v(d(u,v)-r))=\varnothing$ for a specifiic $r\in[0,d(u,v)]$ is at most
\[\prod_{w\in B_u(r)\cup B_v(d(u,v)-r)}(1-\min\{\frac{6\ln n}{\varepsilon}\cdot\|\vec{x}_w^*\|^2,1\})\le e^{-\sum_{w\in B_{u}(v)}\frac{6\ln n}{\varepsilon}\cdot\|\vec{x}_w^*\|^2}\le\frac{1}{n^2}\]
\end{proof}

By using union bound over all the different $r$ we used in our SDP, the probability that there exists an $r\in[0,d(u,v)]$ where $A\cap(B_u(r)\cup B_v(d(u,v)-r))=\varnothing$ is at most $\frac{1}{n^2}\cdot n=\frac{1}{n}$, which means $d(u,v)<\min_{w\in A}d(u,w)+\min_{w\in A}d(v,w)-1$ with probability at most $\frac{1}{n}$, so the probability that $\{u,v\}\in R$ is at most $\frac{1}{n}$.

Therefore, let $OPT_{SDP_{PR}}$ denotes the optimal cost of $SDP_{PR}$, then the expected cost of the rounding algorithm is at most
\[\sum_{v\in V}(n-f)\cdot\frac{3\ln n}{\varepsilon}\cdot\|\vec{x}_v^*\|^2+\frac{2}{\varepsilon}\cdot\sum_{u,v\in V}\|\vec{y}_{uv}^*\|^2+n^2\cdot\frac{1}{n}\le O(\log n)\cdot OPT_{SDP_{PR}}+n\le O(\log n)\cdot OPT\]
because $OPT\ge\Omega(n)$, which proves Theorem~\ref{thm:pro}.

\end{document}